\def\VersionFinal{}

\ifdefined\VersionLong%
	\newcommand{\LongVersion}[1]{#1}
	\newcommand{\ShortVersion}[1]{}
\else
	\newcommand{\LongVersion}[1]{}
	\newcommand{\ShortVersion}[1]{#1}
\fi

\LongVersion{%
\documentclass[onecolumn,journal]{IEEEtran}
}
\ShortVersion{%
\documentclass[lettersize,journal]{IEEEtran}
}
\usepackage{array}
\usepackage[caption=false]{subfig}
\usepackage{textcomp}
\usepackage{stfloats}
\usepackage{url}
\usepackage{verbatim}
\usepackage{graphicx}
\usepackage{cite}
\usepackage{siunitx}

\ifdefined\VersionWithComments%
	\usepackage[colorinlistoftodos,textsize=footnotesize]{todonotes}
	\setuptodonotes{inline}
\else
	\usepackage[disable]{todonotes}
\fi
\newcommand{\wantcite}[1][]{\todo{cite\IfBlankTF{#1}{}{: #1}}}

\newcommand{\gennote}[3]{\todo[linecolor=#2,backgroundcolor=#2!25,bordercolor=#2]{#3: #1}}

\newcommand{\mw}[1]{\gennote{#1}{orange}{MW}}

\newcommand{\instructions}[1]{{\gennote{\bfseries #1}{red}{Instructions}}}
\newcommand{\reviewer}[2]{{\gennote{``#2''}{purple}{Reviewer #1}}}
\newcommand{\revision}[1]{{#1}}
\newcommand{\finalversion}[1]{{#1}}

\usepackage[utf8]{inputenc}
\usepackage[english]{babel} %
\usepackage{csquotes}
\usepackage[normalem]{ulem}

\usepackage{setspace}

\newcommand{\Continue}{\textbf{continue}}

\usepackage[ruled,lined,linesnumbered,noend]{algorithm2e}
\SetKwInOut{Input}{input}
\SetKwInOut{Output}{output}
\SetKw{KwPush}{push}
\SetKw{KwPop}{pop}
\SetKw{KwLet}{let}
\SetKw{KwBe}{be}
\SetKw{KwWith}{with}
\SetKw{KwFrom}{from}
\SetKw{KwCompute}{compute}
\SetKw{KwOr}{or}
\SetKw{KwAnd}{and}
\SetKwProg{Fn}{Function}{:}{}

\makeatletter
\renewcommand{\nllabel}[1]
 {{\let\@currentlabel\algocf@currentlabel
  \let\@currentcounter\algocf@currentcounter
  \label{#1}}}%

\renewcommand{\algocf@nl@sethref}[1]{%
  \renewcommand{\theHAlgoLine}{\thealgocfproc.#1}%
  \hyper@refstepcounter{AlgoLine}%
  \gdef\algocf@currentlabel{#1}%
  \gdef\algocf@currentcounter{AlgoLine}%
 }%
\makeatother

\usepackage{paralist} %
\usepackage{xspace}

\newenvironment{ienumeration}
	{\ifdefined\VersionLong\begin{enumerate}\else\begin{inparaenum}[\itshape i\upshape)]\fi}
	{\ifdefined\VersionLong\end{enumerate}\else\end{inparaenum}\fi}

\usepackage{amsthm} %
\usepackage{thmtools} %
\usepackage{amsmath} %
\usepackage{amssymb} %
\usepackage{mathtools} %
\usepackage{multirow}
\usepackage{amsfonts}

\usepackage{listings}
\usepackage{color}

\usepackage{mathpartir}
\usepackage{tikz}
\usetikzlibrary{automata, positioning, arrows.meta}
\usetikzlibrary{decorations.pathmorphing}

\usepackage[framemethod=tikz]{mdframed}

	\definecolor{mygreen}{rgb}{0,0.6,0}
	\definecolor{mygray}{rgb}{0.5,0.5,0.5}
	\definecolor{mymauve}{rgb}{0.58,0,0.82}
	\definecolor{weborange}{RGB}{255,165,0}

\lstdefinestyle{log}{
	backgroundcolor=\color{white},   %
	basicstyle=\scriptsize,        %
	breakatwhitespace=false,         %
	breaklines=true,                 %
	captionpos=b,                    %
	commentstyle=\color{mygreen},    %
	deletekeywords={...},            %
	escapeinside={\%*}{*)},          %
	extendedchars=true,              %
	frame=single,	                   %
	keepspaces=true,                 %
	keywordstyle=\color{red!70!black}\bfseries,       %
	morekeywords={@, open, close, update},            %
	numbers=left,                    %
	numbersep=5pt,                   %
	numberstyle=\tiny\color{mygray}, %
	rulecolor=\color{black},         %
	showspaces=false,                %
	showstringspaces=false,          %
	showtabs=false,                  %
	stepnumber=1,                    %
	stringstyle=\color{mymauve},     %
	tabsize=2,	                   %
	classoffset=1, %
	otherkeywords={@},
	morekeywords={@},
	keywordstyle=\color{weborange},
	classoffset=0,
}

\ifdefined\VersionWithComments
	\usepackage{draftwatermark}
	\SetWatermarkText{draft}
	\SetWatermarkScale{3}
	\SetWatermarkColor[gray]{0.9}
\fi

\usepackage[svgnames,table]{xcolor}
\definecolor{USPNcobalt}{HTML}{293358}
\definecolor{USPNocre}{HTML}{8b7d6d}
\definecolor{USPNblanc}{HTML}{ffffff}
\definecolor{USPNceruleen}{HTML}{354878}
\definecolor{USPNsable}{HTML}{ad947e}

\usepackage[
\ifdefined\VersionFinal%
		pdfauthor={Tsubasa~Matsumoto, Kazuki~Watanabe, and Masaki~Waga},
\else%
		pdfauthor={Anonymous Authors},%
\fi
		pdftitle={Specification-Guided Path Shortcutting for Efficient Probabilistic Model Checking},
		breaklinks  = true,
		colorlinks  = true,
	\ifdefined \VersionWithComments
	\fi
		citecolor   = USPNsable,
		linkcolor   = USPNocre,
		urlcolor    = USPNceruleen,
	]{hyperref}

\usepackage[capitalise,english,nameinlink]{cleveref} %
\crefname{line}{\text{line}}{\text{lines}} %
\crefname{assumption}{\text{Assumption}}{\text{Assumptions}} %

\usepackage{wrapfig}
\usepackage{tikz}
\usetikzlibrary{arrows,automata,positioning,math}
\tikzstyle{every node}=[initial text=]
\tikzstyle{final}=[double]
\tikzstyle{accepting}=[final]

\usepackage{booktabs}

\newcommand{\setN}{{\mathbb N}}

\newcommand{\setQ}{{\mathbb Q}}

\newcommand{\init}{_{\mathit{init}}}

\newcommand{\powerset}[1]{2^{#1}}
\newcommand{\emptyword}{\varepsilon}

\newcommand{\MC}{\mathcal{M}}
\newcommand{\state}{\ensuremath{s}}
\newcommand{\initialState}{\ensuremath{s\init}}
\newcommand{\States}{S}
\newcommand{\edge}{e}
\newcommand{\Edges}{E}
\newcommand{\InEdges}[1]{\mathit{IE}_{#1}}
\newcommand{\OutEdges}[1]{\mathit{OE}_{#1}}
\newcommand{\MCTransition}{\delta}

\newcommand{\pathvar}{p}
\newcommand{\infpathvar}{p^{\omega}}
\newcommand{\trace}{t}
\newcommand{\run}{r}

\newcommand{\automaton}{\mathcal{A}}

\newcommand{\position}{q}
\newcommand{\initialPosition}{q\init}
\newcommand{\Positions}{Q}
\newcommand{\DRATransition}{\Delta}
\newcommand{\AccCondition}{\mathit{Acc}}
\newcommand{\reclang}[1]{\mathcal{L}(#1)}

\newcommand{\MCCS}{T}

\newcommand{\MCCSCandidate}{\tilde{\MCCS}}

\newcommand{\reducedMC}{\MC'}
\newcommand{\reducedMCFinal}[1][\automaton]{{\MC/{\sim_{#1}}}}
\newcommand{\MCCSBound}{K}

\newcommand{\formula}{\varphi}
\newcommand{\fml}{\varphi}
\newcommand{\Eventually}{\ensuremath{\Diamond}}
\newcommand{\Globally}{\ensuremath{\square}}
\newcommand{\Next}{\ensuremath{\mathcal{X}}}
\newcommand{\Until}{\ensuremath{\mathbin{\mathcal{U}}}}

\usepackage{pifont}%

\definecolor{vertfonce}{rgb}{0.0, 0.5, 0.0}
\definecolor{rougefonce}{rgb}{1, 0.0, 0.0}

\theoremstyle{plain}
\newtheorem{lemma}{Lemma}
\newtheorem{proposition}[lemma]{Proposition}
\newtheorem{theorem}[lemma]{Theorem}

\theoremstyle{definition}
\newtheorem{definition}[lemma]{Definition}
\newtheorem{example}[lemma]{Example}
\theoremstyle{remark}

\usepackage{verbatim} %

\ifdefined\VersionFinal
\else
\usepackage[switch]{lineno} %
\linenumbers
\fi

\newcommand{\ourTool}{\textsf{Storm-SGPS}} %
\newcommand{\storm}{\textsf{Storm}}

\newcommand{\brp}{\textsf{BRP}}
\newcommand{\crowds}{\textsf{CROWDS}}
\newcommand{\egl}{\textsf{EGL}}
\newcommand{\haddad}{\textsf{HM}}

\newcommand{\leader}{\textsf{Leader}}
\newcommand{\nand}{\textsf{NAND}}

\usepackage{colortbl}
\newcommand{\tbcolor}{\cellcolor{green!25}\bf}

\newcommand{\eg}{e.g.,\xspace}

\newcommand{\ie}{i.e.,\xspace}

\newcommand{\wrt}{w.r.t.\xspace}

\newcommand{\AP}{\mathit{AP}}
\newcommand{\ap}{\mathit{c}}
\newcommand{\transition}{\delta}
\newcommand{\etrace}[1]{\mathrm{P}_{#1}}
\newcommand{\eotrace}[1]{\mathrm{P}^{\omega}_{#1}}

\newcommand{\oalpha}[1]{\mathrm{O}(#1)}

\newcommand{\replacedMC}[2]{\mathbb{R}(#1, #2)}

\newcommand{\coverSuffixes}[2]{\mathit{MCCS}(#1, #2)}
\newcommand{\replaceC}[2]{#1(#2)}
\newcommand{\initpath}{\mathrm{head}}
\newcommand{\last}{\mathrm{last}}
\newcommand{\prob}{\mathbb{P}}
\newcommand{\erase}[2]{\mathbb{E}(#1, #2)}
\newcommand{\cyl}{\mathrm{Cyl}}
\newcommand{\sigmaAlg}{\mathcal{F}}
\newcommand{\rpr}{\mathrm{RPr}}

\tikzstyle{rqanswer} = [
 draw=black,
 fill=gray!30,
 text=black,
 line width=0.5pt,
  text width = \linewidth - 1.6 ex - 1pt,
  inner sep = 0.8 ex,
  rounded corners=4pt]
\newcommand{\rqanswer}[2]{%
	\smallskip

	\noindent%
	\begin{tikzpicture}%
	\draw node[rqanswer]{\textbf{Answer to {#1}}:{ #2}};%
	\end{tikzpicture}%
}

\newcommand{\defProblem}[3]
{%
	\smallskip

	\noindent%
	\begin{tikzpicture}%
	\draw node[rqanswer]{
		\small%
		\textbf{#1 Problem:}\\
		\textsc{Input}: #2\\
		\textsc{Problem}: #3
	};%
	\end{tikzpicture}%
	\smallskip
}

\usepackage{placeins}
\usepackage{comment}
\ifdefined\VersionLong%
        \includecomment{LongVersionBlock}
        \excludecomment{ShortVersionBlock}
\else
        \excludecomment{LongVersionBlock}
        \includecomment{ShortVersionBlock}
\fi
\begin{document}

\title{Specification-Guided \revision{Path Shortcutting}\\ for Efficient Probabilistic Model Checking}

\ifdefined\VersionFinal%
\author{Tsubasa~Matsumoto,
Kazuki~Watanabe, and
Masaki~Waga
\thanks{This work is \LongVersion{partially }supported by JST BOOST Grant No.\ JPMJBY24H8, JST PRESTO Grant No.\ JPMJPR22CA \& JPMJPR25KD, JST CREST Grant No.\ JPMJCR26X4 \& JPMJCR2012, and JSPS KAKENHI Grant No.\ 25H01113.}%
\thanks{T.\ Matsumoto and M.\ Waga are with Kyoto University, Kyoto, Japan. K.\ Watanabe and M.\ Waga are with National Institute of Informatics, Tokyo, Japan.}%
\LongVersion{%
\thanks{This is the author (and extended) version of the manuscript of the same name published in IEEE Transactions on Computer-Aided Design of Integrated Circuits and Systems (TCAD).
The final version is available at \url{https://ieeexplore.ieee.org/}.}}
}
\else
\author{Anonymous Authors}
\fi

\LongVersion{%
  \date{}
}

\ShortVersion{%
\markboth{Journal of \LaTeX\ Class Files,~Vol.~14, No.~8, August~2021}%
{Shell \MakeLowercase{\textit{et al.}}: A Sample Article Using IEEEtran.cls for IEEE Journals}
}

\maketitle

\LongVersion{
	\thispagestyle{plain}
}

\FloatBarrier
\begin{abstract}
 Given the safety-critical nature of many embedded systems, their safety assurance is essential.
 Because such systems are typically stochastic, \emph{probabilistic model checking} is a particularly important technique.
 However, there is a well-known scalability issue due to state-space explosion, especially when verifying complex properties.
 To mitigate this issue, we propose \emph{specification-guided \revision{path shortcutting}} for probabilistic systems, focusing on \emph{Markov chains (MCs)} and \emph{$\omega$-regular properties}.
 The key idea is that, when the verified property is fixed, certain sequences of transitions in an MC can be replaced with a single transition without changing the satisfaction probability, and thus, we can reduce the state space of the MC.\@
 We implement the proposed \revision{path shortcutting} and evaluate its contribution to the performance of probabilistic model checking, using \storm{} as the baseline model checker.
 The results suggest that our approach often outperforms the baseline, particularly on benchmark instances with complex specifications.
\end{abstract}
\section{Introduction}\label{section:introduction}

\subsection{Embedded Systems and Stochastic Systems}

\IEEEPARstart{E}{nsuring} the correctness of embedded systems is essential due to their irreplaceable role, yet it remains notoriously challenging.
A major difficulty arises from their stochastic behavior, which is commonly observed in network protocols and IoT devices.
For instance, communication between servers may fail for unpredictable reasons related to physical conditions on devices; such stochastic behavior can be modeled using stochastic systems.
\emph{Markov chains (MCs)} are among the simplest and most widely used probabilistic models for representing such systems, with applications ranging from network protocols to reliable computing~\cite{KwiatkowskaNP12,BaierKatoen}.

\begin{example}%
 \label{example:running_example}
	\cref{fig:motivatingMC} illustrates an example of an MC. Each state has outgoing transitions labeled with atomic propositions and associated probabilities.
	For instance, from state $\state_1$, there is a transition to $\state_3$ labeled with $\{a\}$ and probability $0.1$.
\end{example}

\begin{figure}[t]
	\center
	\scalebox{\ShortVersion{0.75}\LongVersion{0.9}}{
\begin{tikzpicture}[
  ->,                    %
  >=Stealth,             %
  shorten >=1pt,         %
  auto,                  %
  node distance=2.5cm,   %
  semithick,              %
  every state/.style={
    font=\small,
    minimum size=6mm
  },scale=0.95,every node/.style={initial text={},transform shape}
]
  \node[state, initial] (q0) at (0, 0) {$\state_0$};
  \node[state] (q1) at (2, 0) {$\state_1$};
  \node[state] (q2) at (4, 1.5) {$\state_2$};
  \node[state] (q3) at (5, 0) {$\state_3$};
  \node[state] (q4) at (3, -1.5) {$\state_4$};
  \node[state] (q5) at (1, -1.5) {$\state_5$};

  \path
    (q0) edge node {$\{a\}$} (q1)
	(q1) edge node[xshift=2mm] {$\{b\}, 0.5$} (q2)
	(q2) edge[loop right] node {$\{a\}$} ()
	(q1) edge node[yshift=-0.5mm] {$\{a\}, 0.1$} (q3)
	(q3) edge[loop right] node {$\{b\}$} ()
	(q1) edge node[yshift=-2mm] {$\{a\},0.4$} (q4)
	(q4) edge node {$\{b\}$} (q5)
	(q5) edge node[yshift=-2mm] {$\{a\}, 0.3$} (q1)
	(q5) edge[loop left] node {$\{a\}, 0.7$} ();

\end{tikzpicture}
}
\caption{
	An MC. We omit the transition probability if it is $1$.
}
	\label{fig:motivatingMC}
\end{figure}
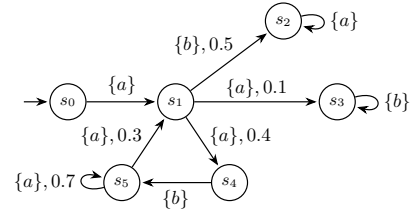

\subsection{Probabilistic Model Checking}

A powerful methodology for verifying embedded systems is \emph{probabilistic model checking}, which has been actively developed and applied to the verification of embedded systems for decades~\cite{BaierKatoen,NormanPKSG05,KatoenW16,KwiatkowskaNP10}.
Verifying linear-time properties of MCs is a common verification problem in probabilistic model checking and is supported by existing tools such as PRISM~\cite{KwiatkowskaNP11} and \storm{}~\cite{HenselJKQV22}.
Formally, given an MC $\MC$ and a \emph{linear-time temporal logic (LTL)} formula $\formula$, probabilistic model checking concerns the probability $\prob(\MC \models \formula)$ that infinite paths of $\MC$ satisfy $\formula$.

\subsection{Existing Approach and Challenge}

A long-standing bottleneck in probabilistic model checking stems from the size of the state space.
When the verified MC has a large state space, a model-checking algorithm requires more iterations to compute the probability that $\fml$ is satisfied, and thus, the verification process takes more time.
Moreover, recent studies have shown that the performance of modern probabilistic model checkers degrades significantly when the representation of an MC does not fit into memory due to the explosion of the state space (\eg{}~\cite{Katoen16,Hensel18,JungesS22,WatanabeEAH23}).

To mitigate this issue, \emph{bisimulation minimization} has been employed as an effective abstraction technique in state-of-the-art model checkers such as \storm{}~\cite{HenselJKQV22,Hensel18}.
Specifically, it constructs a \emph{quotient MC} $\MC/_{\simeq}$ by identifying \emph{bisimilar} states $\state_1 \simeq \state_2$ and merging states that are bisimilar to each other into a single state.
This preprocessing step can reduce the number of states before model checking and is sound in the following sense:
given an MC $\MC$, we have
\begin{equation}
	\label{eq:bisimulation}
	\prob(\MC\models \varphi) = \prob(\MC/_{\simeq}\models \varphi) \quad \text{for any LTL formula } \varphi.
\end{equation}
See~\cite{BaierKatoen} for details\footnote{This equivalence holds for PCTL* formulas, which subsume LTL. In fact, PCTL* characterizes probabilistic bisimulation.}.
By~\cref{eq:bisimulation}, it suffices to verify the quotient MC $\MC/_{\simeq}$ instead of the original MC $\MC$, while potentially reducing the state space.

However, bisimulation is often too strong a requirement to hold in realistic MCs.
In fact, in~\cref{fig:motivatingMC}, no two distinct states are bisimilar.

\subsection{Our Approach}
In this paper, we propose a novel abstraction method, called \emph{specification-guided \revision{path shortcutting}}, for MCs with respect to a given $\omega$-regular property $\varphi$, including properties expressible in LTL.\@
More specifically, given an MC $\MC$ and an $\omega$-regular property $\varphi$, our abstraction constructs an MC $\MC'$ that satisfies the following equivalence:
\begin{equation}
	\label{eq:ourabstraction}
	\prob(\MC\models \varphi) = \prob(\MC'\models \varphi).
\end{equation}
Importantly, in~\cref{eq:ourabstraction}, the equivalence is guaranteed to hold for the given property $\varphi$, in contrast to~\cref{eq:bisimulation}, which holds for any LTL formula.
This dependence on the specification enables us to eliminate states that are irrelevant to model checking against $\varphi$.
For instance, our method can reduce the MC in \cref{fig:motivatingMC} to an MC $\MC'$ with four states by eliminating two states; we demonstrate this procedure in~\cref{sec:overview} in detail.
Our abstraction can be used as a preprocessing step before model checking, similar to bisimulation minimization.
We provide an overview of the workflow in~\cref{fig:outline}.
\begin{figure}
	\center
  \scalebox{\ShortVersion{0.8}\LongVersion{0.9}}{
	\begin{tikzpicture}
    \node[draw] at (0.45, -3) (product) {an MC $\MC$ and an LTL formula $\formula$};
    \node[draw] at (0.45, -4.5) (simpleproduct) {an MC $\MC'$ and an LTL formula $\formula$};
	\node[draw] at (0.45, -6.0) (modelcheck) {Compute the probability $\prob(\MC'\models \formula)$};
    \draw[->, very thick, align=left] (product) -- node[midway, right] {Our novel specification-guided \revision{path shortcutting}  } (simpleproduct);
	\draw[->, very thick, align=left] (simpleproduct) -- node[midway, right] {Running a probabilistic model checking algorithm\\ with existing solvers (including \storm{})} (modelcheck);
	\end{tikzpicture}
  }
	\caption{Workflow of our proposed probabilistic model checking procedure.}
	\label{fig:outline}
\end{figure}
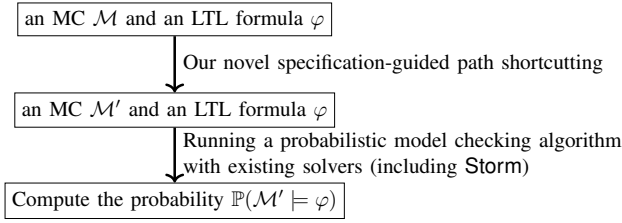

We implement \revision{\ourTool{}, a prototype tool for} specification-guided \revision{path shortcutting}, and evaluate its performance using the workflow shown in~\cref{fig:outline}, comparing against the state-of-the-art model checker \storm{}~\cite{HenselJKQV22}.
Our experimental results demonstrate that our \revision{workflow} generally \revision{improves the end-to-end performance} on many benchmark instances consisting of MCs from the QComp benchmark suite~\cite{HartmannsKPQR19}, achieving up to an \revision{18}$\times$ speedup.%

\subsection{Contributions and Outline}
We make the following contributions.
\begin{itemize}
	\item We introduce a notion of \emph{edge replaceability} for MCs with respect to a given $\omega$-regular property (\cref{section:wordlevelSGA}).
	\item We provide a sufficient condition for \emph{state elimination} based on edge replaceability (\cref{sec:eraseStates}).
	\item We present a novel abstraction method, called \emph{specification-guided \revision{path shortcutting}} (\cref{section:algorithm}).
	\item We empirically evaluate the effectiveness of our approach by comparing it with the model checker \storm{}~\cite{HenselJKQV22} (\cref{section:experiments}).
\end{itemize}
Before presenting these contributions, we illustrate our specification-guided \revision{path shortcutting} using the example MC in~\cref{fig:motivatingMC} (\cref{sec:overview}), and recall preliminaries on probabilistic model checking (\cref{section:preliminaries}).

\section{Overview}
\label{sec:overview}

We illustrate our specification-guided \revision{path shortcutting} using the MC $\MC$ in~\cref{fig:motivatingMC}.
We consider the LTL formula $\formula \coloneqq \Globally(a \implies \Eventually b)$.
First, we construct a deterministic Rabin automaton $\automaton_{\formula}$ recognizing the set of $\omega$-words satisfying $\formula$;~\cref{fig:motivatingDRA} depicts $\automaton_{\formula}$.

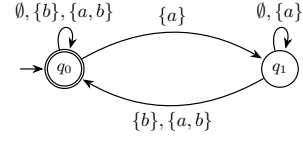
\begin{figure}[t]
	\center
	\scalebox{\ShortVersion{0.75}\LongVersion{0.9}}{
	\begin{tikzpicture}[
  ->,                    %
  >=Stealth,             %
  shorten >=1pt,         %
  auto,                  %
  node distance=2.5cm,   %
  semithick,              %
  every state/.style={
    font=\small,
    minimum size=6mm
  },scale=0.95,every node/.style={initial text={},transform shape}
]
  \node[state, initial,accepting] (q0) at (0, 0) {$\position_0$};
  \node[state] (q1) at (4, 0) {$\position_1$};

  \path
	(q0) edge[bend left] node {$\{a\}$} (q1)
	(q1) edge[bend left] node {$\{b\}, \{a, b\}$} (q0)
    (q0)[loop above] edge node {$\emptyset, \{b\}, \{a, b\}$} ()
	(q1)[loop above] edge node {$\emptyset, \{a\}$} ();
\end{tikzpicture}
}
\caption{
	Deterministic Rabin automaton (DRA) equivalent to the LTL formula $\Globally(a \implies \Eventually b)$. An $\omega$-word is accepted by this DRA if the run visits $\position_0$ infinitely often.
}
	\label{fig:motivatingDRA}
\end{figure}

\begin{figure}[t]
	\center
\scalebox{\ShortVersion{0.8}\LongVersion{0.9}}{
\begin{tikzpicture}[
  ->,                    %
  >=Stealth,             %
  shorten >=1pt,         %
  auto,                  %
  node distance=2.5cm,   %
  semithick,              %
  every state/.style={
    font=\small,
    minimum size=6mm
  },
  unreachable/.style={
    state,
    draw=gray,
    text=gray,
    dashed,
    fill=gray!15
  },scale=0.95,every node/.style={initial text={},transform shape}
]
  \node[state, initial] (q0) at (0, 0) {$\state_0$};
  \node[state] (q1) at (2, 0) {$\state_1$};
  \node[state] (q2) at (4, 1.5) {$\state_2$};
  \node[state] (q3) at (5, 0) {$\state_3$};
  \node[unreachable] (q4) at (3, -1.5) {$\state_4$};
  \node[state] (q5) at (1, -1.5) {$\state_5$};

  \path
    (q0) edge node {$\{a\}$} (q1)
	(q1) edge node[xshift=2mm] {$\{b\}, 0.5$} (q2)
	(q2) edge[loop right] node {$\{a\}$} ()
	(q1) edge node[yshift=-0.5mm] {$\{a\}, 0.1$} (q3)
	(q3) edge[loop right] node {$\{b\}$} ()
	(q1) edge[color=blue,bend left] node[yshift=5mm] {\textcolor{blue}{\uline{$\{b\},0.4$}}} (q5)
	(q4) edge node {$\{b\}$} (q5)
	(q5) edge[bend left] node[yshift=-4mm] {\textcolor{red}{\uline{$\{a\},0.3$}}} (q1)
	(q5) edge[loop left] node {$\{a\}, 0.7$} ();

\end{tikzpicture}
}
\caption{
	The MC obtained by replacing the edge from $\state_1$ to $\state_4$ with the new edge from $\state_1$ to $\state_5$.
}
\label{fig:step2MC}
\end{figure}
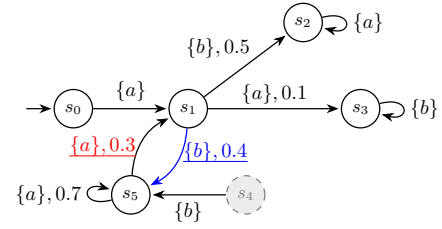

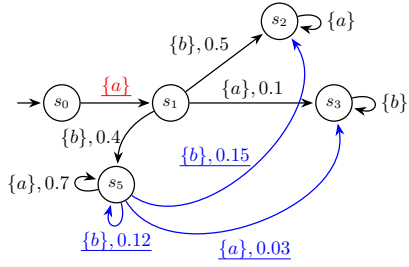
\begin{figure}[t]
	\center
\scalebox{0.8}{
\begin{tikzpicture}[
  ->,                    %
  >=Stealth,             %
  shorten >=1pt,         %
  auto,                  %
  node distance=2.5cm,   %
  semithick,              %
  every state/.style={
    font=\small,
    minimum size=6mm
  },
  unreachable/.style={
    state,
    draw=gray,
    text=gray,
    dashed,
    fill=gray!15
  },scale=0.90,every node/.style={initial text={},transform shape}
]
  \node[state, initial] (q0) at (0, 0) {$\state_0$};
  \node[state] (q1) at (2, 0) {$\state_1$};
  \node[state] (q2) at (4, 1.5) {$\state_2$};
  \node[state] (q3) at (5, 0) {$\state_3$};
  \node[state] (q5) at (1, -1.5) {$\state_5$};

  \path
    (q0) edge node {\textcolor{red}{\uline{$\{a\}$}}} (q1)
	(q1) edge node[xshift=2mm] {$\{b\}, 0.5$} (q2)
	(q2) edge[loop right] node {$\{a\}$} ()
	(q1) edge node[yshift=-0.5mm] {$\{a\}, 0.1$} (q3)
	(q3) edge[loop right] node {$\{b\}$} ()
	(q1) edge[bend right] node[left,yshift=-1mm] {$\{b\},0.4$} (q5)
	(q5) edge[loop left] node {$\{a\}, 0.7$} ()
	(q5) edge[color=blue,loop below] node {\textcolor{blue}{\uline{$\{b\},0.12$}}} ()
	(q5) edge[color=blue,bend right=80] node[below,yshift=-2mm] {\textcolor{blue}{\uline{$\{a\},0.03$}}} (q3)
	(q5) edge[color=blue,bend right=80] node[left,yshift=1mm] {\textcolor{blue}{\uline{$\{b\},0.15$}}} (q2);
\end{tikzpicture}
}
\caption{
	The MC obtained by removing the edge from $\state_5$ to $\state_1$.
}
\label{fig:step3MC}
\end{figure}

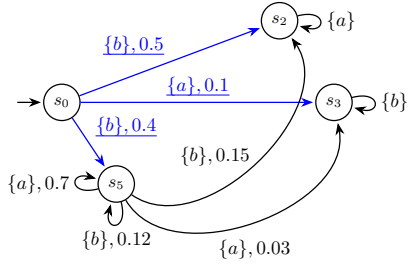
\begin{figure}[t]
	\center
\scalebox{0.8}{
\begin{tikzpicture}[
  ->,                    %
  >=Stealth,             %
  shorten >=1pt,         %
  auto,                  %
  node distance=2.5cm,   %
  semithick,              %
  every state/.style={
    font=\small,
    minimum size=6mm
  },
  unreachable/.style={
    state,
    draw=gray,
    text=gray,
    dashed,
    fill=gray!15
  },scale=0.90,every node/.style={initial text={},transform shape}
]
  \node[state, initial] (q0) at (0, 0) {$\state_0$};
  \node[state] (q2) at (4, 1.5) {$\state_2$};
  \node[state] (q3) at (5, 0) {$\state_3$};
  \node[state] (q5) at (1, -1.5) {$\state_5$};

  \path
    (q0) edge[color=blue] node {\textcolor{blue}{\uline{$\{b\}, 0.5$}}} (q2)
	(q0) edge[color=blue] node {\textcolor{blue}{\uline{$\{a\}, 0.1$}}} (q3)
	(q2) edge[loop right] node {$\{a\}$} ()
	(q3) edge[loop right] node {$\{b\}$} ()
	(q0) edge[color=blue] node {\textcolor{blue}{\uline{$\{b\}, 0.4$}}} (q5)
	(q5) edge[loop left] node {$\{a\}, 0.7$} ()
	(q5) edge[loop below] node {$\{b\},0.12$} ()
	(q5) edge[bend right=80] node[below,yshift=-2mm] {$\{a\},0.03$} (q3)
	(q5) edge[bend right=80] node[left,yshift=1mm] {$\{b\},0.15$} (q2);
\end{tikzpicture}
}
\caption{
	The MC obtained by removing the edge from $\state_0$ to $\state_1$, and erasing the state $\state_1$.
}
	\label{fig:step4MC}
\end{figure}

By examining $\automaton_{\formula}$, we observe that reading the word $\{a\}\{b\}$ is equivalent to reading the character $\{b\}$.
This is because, from any position in $\automaton_{\formula}$,
reading both $\{a\}\{b\}$ and $\{b\}$ leads to $\position_0$, and
the position $\position_0$, which is relevant for acceptance, is visited only in the last step.
Therefore, we can \emph{replace} the two-step transition from $\state_1$ to $\state_5$ via $\state_4$ labeled with $\{a\}\{b\}$ with a single-step transition from $\state_1$ to $\state_5$ labeled $\{b\}$ with probability $0.4$, which is the probability of the original two-step transition.
This replacement can be realized by removing the edge from $\state_1$ to $\state_4$ and adding a new edge from $\state_1$ to $\state_5$ labeled $\{b\}$ with probability $0.4$.
The resulting MC is shown in~\cref{fig:step2MC}, where $\state_4$ is \emph{eliminated} as it is no longer reachable from the initial state $\state_0$.

Next, consider the edge from $\state_5$ to $\state_1$.
By a similar argument, the word $\{a\}\{a\}$ is equivalent to the character $\{a\}$ with respect to $\automaton_{\formula}$.
The two equivalences allow us to replace the edge from $\state_5$ to $\state_1$ with the three edges highlighted in blue in~\cref{fig:step3MC}.
Again, the probabilities of the new edges are obtained by multiplication; for example, the edge from $\state_5$ to $\state_2$ has probability $0.3 \times 0.5$.
Notably, we shrink the cycle $\state_1 \cdot \state_4 \cdot \state_5$ in $\MC$ into the single state $\state_5$ by creating the new self-loop in~\cref{fig:step3MC}.
This is surprising because the cycle is not a bottom strongly connected component, which is the class of components that can typically be collapsed into a single state~\cite{BaierKatoen}.

Finally, we obtain the MC shown in~\cref{fig:step4MC} by replacing the edge from $\state_0$ to $\state_1$ with three new edges from $\state_0$, and then eliminating the state $\state_1$ that is no longer reachable from the initial state $\state_0$.
Overall, this results in an equivalent MC with four states, reduced from the original MC $\MC$ with six states, with respect to the given specification $\formula$.
This highlights that our abstraction goes beyond classical bisimulation-based reductions, which do not eliminate any states in $\MC$.
In this paper, we formally present this abstraction in detail and demonstrate its effectiveness compared to a state-of-the-art model checker \storm{}~\cite{HenselJKQV22}.

\section{Preliminaries}\label{section:preliminaries}

We recall the preliminaries of probabilistic model checking and formally introduce our problem of interest, namely probabilistic model checking of Markov chains for quantitative $\omega$-regular properties.

We denote the sets of rational and natural numbers by $\setQ$ and $\setN$, respectively.
For a set $X$, we denote its power set by $\powerset{X}$.
For a set $X$, a \emph{word} over $X$ is a finite sequence of elements in $X$, and we write $X^{\ast}$ for the set of words over $X$.
We use $\emptyword$ to denote the empty word.
We let $X^{+} \coloneqq X^{\ast} \setminus \{\emptyword\}$.
An \emph{$\omega$-word} over $X$ is an infinite sequence of elements in $X$, and we write $X^{\omega}$ for the set of $\omega$-words over $X$.
\revision{Throughout this paper, for a finite set $\AP$ of atomic propositions, we write $\Sigma$ for $\powerset{\AP}$. We call $\ap \in \Sigma$ a \emph{label}.}

\subsection{Markov chain}
\begin{definition}[Markov chain]
  A \emph{(transition-labeled) Markov chain} (MC) is a 4-tuple $(\States, \AP, \MCTransition, \initialState)$, where $\States$ is a finite set of \emph{states},
  $\AP$ is a finite set of \emph{atomic propositions}, $\MCTransition\colon \States \times \powerset{\AP} \times \States\rightarrow [0, 1]\cap \setQ$ is a \emph{transition probability}, that is, $\sum_{(\ap, \state_2)\in \powerset{\AP}\times \States}\MCTransition(\state_1, \ap, \state_2) = 1$ for each $\state_1\in \States$, and $\initialState$ is a fixed initial state.
\end{definition}
Note that we assign a label $\ap \in \Sigma$ to each \emph{edge} in an MC.
Formally, an \emph{edge} $\edge$ of an MC $\MC = (\States, \AP, \MCTransition, \initialState)$ is defined as a triple $\edge = (\state_1, \ap, \state_2) \in \States \times \Sigma \times \States$ such that $\MCTransition(\state_1,\ap, \state_2) > 0$.
We denote the set of all edges of $\MC$ by $\Edges_{\MC}$.

A (finite) \emph{path} $\pathvar$ is a finite sequence of edges $\pathvar = (\state_{i}, \ap_i, \state'_i)_{i\in I} \in \Edges_{\MC}^{\ast}$ such that $\state'_i = \state_{i+1}$.
We write $\initpath(\pathvar)$ and $\last(\pathvar)$ for the first and last state in $\pathvar$.
An \emph{infinite path} $\infpathvar$ is an infinite sequence of edges $\infpathvar \in \Edges_{\MC}^{\omega}$.
We write $\mathrm{P}_{\MC}$ and  $\mathrm{P}^{\omega}_{\MC}$ for the set of finite and infinite paths in $\MC$, respectively.
Given a state $\state$, we also write $\mathrm{P}_{\MC}(\state)$ and  $\mathrm{P}^{\omega}_{\MC}(\state)$ for the set of finite and infinite paths starting from $\state$, respectively.

For a path $\pathvar$, we define its \emph{trace} $\oalpha{\pathvar} \in \Sigma^{\ast}$ as the sequence of labels along $\pathvar$.
The trace $\oalpha{\infpathvar}$ of an infinite path $\infpathvar$ is defined analogously.

Given an edge $\edge = (\state_1, \ap, \state_2)$, we write $\etrace{\edge}$ for the set of paths $\pathvar$ such that $e$ is a prefix of $\pathvar$.
We write $\eotrace{e}$ for the set of infinite paths $\infpathvar$ in $\MC$ such that $\edge$ is a prefix of $\infpathvar$.

For each state $\state_1$, we denote the set of all incoming edges of $\state_1$ by $\InEdges{\state_1}$ and the set of all outgoing edges of $\state_1$ by $\OutEdges{\state_1}$,
\ie $\InEdges{\state_1} = \{(\state_2, \ap, \state_1)\in \Edges_{\MC}\mid \MCTransition(\state_2, \ap, \state_1) > 0\}$ and $\OutEdges{\state_1} = \{(\state_1, \ap, \state_2)\in \Edges_{\MC} \mid \MCTransition(\state_1, \ap, \state_2) > 0\}$.

We recall the probability measure $\prob$ over the set of infinite paths $\Edges_{\MC}^{\omega}$, following~\cite{BaierKatoen}.
Given a finite path $\pathvar$, the \emph{cylinder set} $\cyl(\pathvar)$ is the set of infinite paths $\infpathvar$ such that $\pathvar$ is a prefix of $\infpathvar$.
For a set $P$ of finite paths, we write $\cyl(P)$ for the set of infinite paths that have a prefix in $P$.

The \emph{$\sigma$-algebra} \revision{$\sigmaAlg$} of $\MC$ is the smallest $\sigma$-algebra generated by the cylinder sets of all finite paths.
The \emph{probability measure} $\prob$ of $\MC$ is the unique probability measure \revision{on $\sigmaAlg$} such that for all cylinder sets $\cyl(\pathvar)$,
\[
	\prob\big(\cyl(\pathvar)\big) = \prod_{i=0}^{n} \MCTransition(\state_i,\ap_{i+1}, \state_{i+1}),
\]
where $\pathvar = (\state_0, \ap_1, \state_{1}) \cdot (\state_1, \ap_2, \state_{2}) \cdots (\state_n, \ap_{n+1}, \state_{n+1})$.
By a slight abuse of notation, we often write $\prob(\pathvar)$ for $\prob\big(\cyl(\pathvar)\big)$.

Given a path $\pathvar$ and a set $X \subseteq \States$, we say that $\pathvar$ is a \emph{path to $X$} if it ends in a state in $X$ and does not visit $X$ before its last state.
The \emph{reachability probability} $\rpr_{\MC}(X)$ is the total probability of all such paths from the initial state $\initialState$.

\subsection{Deterministic Rabin Automaton}

For specifications, we use \emph{deterministic Rabin automata (DRAs)}\footnote{More precisely, we employ \emph{generalized} Rabin automata, which have the same expressive power as the (plain) Rabin automata.} that recognize $\omega$-regular languages.
See, \eg{}~\cite{KleinB06,EsparzaKS16} for the details of the construction of DRAs from LTL formulas.

\begin{definition}[deterministic Rabin automata]
	A \emph{deterministic Rabin automaton} (DRA) is a 5-tuple $(\Positions, \AP, \DRATransition, \initialPosition, \AccCondition)$, where
	$\Positions$ is a finite set of \emph{positions}, $\AP$ is a finite set of atomic propositions, $\DRATransition\colon \Positions \times \Sigma \to \Positions$ is a \emph{transition function}, $\initialPosition \in \Positions$ is an initial position, and $\AccCondition \subseteq \powerset{\Positions} \times \powerset{\Positions}$ is a (generalized) Rabin acceptance condition.
\end{definition}

Given a position $\position\in Q$ and a sequence $\trace \in \Sigma^{\ast}$ of labels, we define the transition $\DRATransition(\position, \trace)\in \Positions$ recursively: $\DRATransition(\position, \emptyword) \coloneqq \position$ and $\DRATransition(\position, \ap\cdot \trace) \coloneqq \DRATransition\big(\DRATransition(\position, \ap), \trace \big)$.
For an infinite sequence $\trace\in \Sigma^{\omega}$, the \emph{run} $\run(\trace) \in \Positions^{\omega}$ over $\trace$ is the sequence of positions from $\initialPosition$ by applying $\DRATransition$ to each prefix of $\trace$, \ie{} for $\trace = \ap_1\ap_2\cdots$, $\run(\trace) = \position_0\position_1\cdots$ such that $\position_0 = \initialPosition$ and $\position_{i+1} = \DRATransition(\position_i, \ap_{i+1})$ for each $i\in \setN$.
A run $\run\in \Positions^{\omega}$ is \emph{accepting} if there is a pair $(L, U)\in \AccCondition $ such that there is a position $\position\in L$ that occurs infinitely often in $r$, and for any $\position' \in U$, $\position'$ appears in $r$ only finitely many times.
A DRA accepts an infinite sequence $\trace \in \Sigma^{\omega}$ if the run $\run(\trace)$ over $\trace$ is accepting.
The \emph{recognized language} $\reclang{\automaton} \subseteq \Sigma^{\omega}$ of $\automaton$ is the set of infinite sequences $\trace$ accepted by $\automaton$.

\subsection{Probabilistic Model Checking}
We formally state the target problem as follows:

\defProblem{Probabilistic Model Checking}{an MC $\MC$ and a DRA $\automaton$.}{compute the probability of accepting (infinite) traces, that is, compute
\[
\prob\Big(\big\{\infpathvar\in \mathrm{P}^{\omega}_{\MC}(\initialState) \bigm| \oalpha{\infpathvar}\in \reclang{\automaton} \big \}\Big)\in [0, 1].
\]
We write $\prob(\MC\vDash \automaton)$ for the above probability.
Note that the set of accepting paths (\ie{} infinite paths $\infpathvar$ of $\MC$ whose trace $\oalpha{\infpathvar}$ is accepted by $\automaton$) is measurable, therefore its probability is well-defined~\cite{BaierKatoen}. }

We remark that in our experiments shown in~\cref{section:experiments}, we employ the $\epsilon$-approximation problem for the evaluation, that is, computing a lower bound $l$ and an upper bound $u$ of the probability $\prob(\MC\vDash \automaton)$ such that $|l - u|\leq \epsilon$.

\subsection{The Product Construction}
Finally, we recall the de facto standard algorithm for probabilistic model checking for linear-time temporal properties: the product construction.
By constructing the product MC $\MC \otimes \automaton$ of $\MC$ and $\automaton$, we reduce the original problem to a reachability probability problem on $\MC \otimes \automaton$.

	\begin{definition}[product Markov chain]
		For an MC $\MC = (\States, \AP, \MCTransition_{\MC}, \initialState)$ and a DRA $\automaton = (\Positions, \AP, \DRATransition_{\automaton}, \initialPosition, \AccCondition)$, the \emph{product}
		$\MC\otimes \automaton$ of
		$\MC$ and $\automaton$ is an MC $(\States \times \Positions, \AP, \MCTransition_{\MC\otimes\automaton }, (\initialState, \initialPosition))$, where
		$\MCTransition_{\MC\otimes \automaton}\colon (\States\times \Positions)\times \Sigma\times (\States\times \Positions)\rightarrow [0, 1]\cap \setQ$ is the \emph{transition probability} defined by
		\begin{align*}
			\MCTransition_{\MC\otimes\automaton}\big((\state_1, \position_1), \ap, (\state_2, \position_2)\big)
			\coloneqq \MCTransition_{\MC}\big(\state_1, \ap,\state_2\big)
		\end{align*}
		if $\DRATransition_{\automaton}(\position_1, \ap) = \position_2$, and it is $0$ otherwise.
	\end{definition}

	Unlike the product (unlabeled) MC for state-labeled MCs (e.g.~\cite{BaierKatoen}), the product $\MC \otimes \automaton$ is equipped with labels to ensure a bijective correspondence between paths on $\MC$ and those on $\MC \otimes \automaton$ starting from a state $q$. This design choice simplifies the correctness proof of our approach.

In this reduction, \emph{bottom strongly connected components (BSCCs)} play a key role: the model checking problem is reduced to the reachability probability problem of reaching \emph{accepting BSCCs}.
\revision{Roughly speaking, BSCCs are strongly connected components with no outgoing transitions. } 
	\begin{definition}[BSCC]
		Given a product MC $\MC\otimes \automaton$, a non-empty set $X\subseteq \States\times \Positions$ is a \emph{bottom strongly connected component (BSCC)} if the following conditions hold:
		\begin{itemize}
			\item (SCC): for any $(\state_1, \position_1)$ and $(\state_2, \position_2)$ in $X$, there is a path from $(\state_1, \position_1)$ to $(\state_2, \position_2)$.
			\item (bottom): for any $(\state_1, \position_1)\in X$, $\ap \in \Sigma$, and $(\state_2, \position_2)\in \States\times \Positions$, if $\MCTransition_{\MC\otimes \automaton}\big((\state_1, \position_1), \ap, (\state_2, \position_2)\big) > 0$, then $(\state_2, \position_2)\in X$.
		\end{itemize}
		We say that a BSCC $X$ is \emph{accepting} if there exists $(L, U) \in \AccCondition$ such that (i) there is a state $(\state, \position) \in X$ with $\position \in L$, and (ii) for all $(\state, \position) \in X$, we have $\position \notin U$.
\end{definition}

It is known that the limit behavior of MCs is precisely captured with BSCCs~\cite{BaierKatoen}:
\revision{In MCs, we almost surely end up in a BSCC}. 
\begin{lemma}
	\label{lem:asBSCC}
	Let $\MC$ be an MC and $U$ be the union of all BSCCs in $\MC$.
	We have $\rpr_{\MC}(U) = 1$.
\end{lemma}

\begin{proposition}[correctness of products~\cite{BaierKatoen}]
	Given an MC $\MC$ and a DRA $\automaton$, we have the following equality:
	\[
	\prob(\MC\vDash \automaton) = \rpr_{\MC\otimes \automaton}(X),
	\]
	where $X$ is the union of accepting BSCCs.
\end{proposition}
Reachability probabilities on MCs can be computed in polynomial time by solving linear equation systems (see \eg{}~\cite{BaierKatoen}), and the search for efficient algorithms, including various heuristics, has been actively pursued~\cite{ChatterjeeH08,HaddadM18,HartmannsK20,BrazdilCCFKKPU14}.

We adopt this approach in our setting: after applying our proposed specification-guided \revision{path shortcutting} as preprocessing, we reduce the problem to computing reachability probabilities on the product MC and apply an existing algorithm implemented in \storm{}.

\section{Replaceability of Edges}\label{section:wordlevelSGA}

Let $\MC = (\States, \AP, \transition_{\MC}, \initialState)$ be an MC and $\automaton$ be a DRA with Rabin acceptance condition $(L_i, U_i)_{i \in I}$.
To introduce our specification-guided \revision{path shortcutting}, we first define a notion of \emph{replaceability} of edges in $\MC$ with respect to $\automaton$.

\subsection{Replaceable Edges and Replaced MC}

Recall that in the overview (in~\cref{sec:overview}), the equivalence over words (such as $\{a\}\{b\}$ and $\{b\}$) leads to the replacement of edges.
We formally define this equivalence over \emph{acceptance-preserving words}.

\begin{definition}[acceptance-preserving]
	For a DRA with Rabin acceptance condition $(L_i, U_i)_{i \in I}$, a non-empty word (or trace) $\trace\in \Sigma^{+}$ is \emph{acceptance-preserving} if the following condition is satisfied:
	\begin{itemize}
		\item for any $\position\in Q$, any prefix $\trace'\in \Sigma^+$ of $\trace$, and any $X_i \in \{L_i, U_i\}$ for some $i\in I$, if $\DRATransition(\position, \trace')\in X_i$, then $\DRATransition(\position, \trace)\in X_i$.
	\end{itemize}
\end{definition}

\begin{example}
	For the DRA in \cref{fig:motivatingDRA}, words $\{a\}\{b\}$ and $\{a\}\{a\}$ are acceptance-preserving, whereas $\{b\}\{a\}$ is not acceptance-preserving.
\end{example}

\begin{example}
	Let $\automaton$ be a deterministic co-safety automaton, \ie{} a DRA whose Rabin acceptance condition is of the form $\{(L, \emptyset)\}$ and for any $\position \in L$ and $\ap\in \Sigma$, we have $\DRATransition(\position, \ap)\in L$.
    Then, every word $\trace\in \Sigma^{+}$ is acceptance-preserving \wrt $\automaton$.
\end{example}

\begin{definition}[compatibility relation]
We define the \emph{compatibility relation} $\trace_1\sim_{\automaton} \trace_2$ over the set of acceptance-preserving words by
$\trace_1\sim_{\automaton} \trace_2$ if $\DRATransition(\position, \trace_1) = \DRATransition(\position, \trace_2)$ for any $\position\in Q$.
\end{definition}
Clearly, the compatibility relation is an equivalence relation.
Without loss of generality, we assume that for any labels $\ap_1, \ap_2 \in \Sigma$, $\ap_1 \sim_{\automaton} \ap_2$ implies $\ap_1 = \ap_2$ (note that  all labels are acceptance-preserving).
This assumption is justified because such distinct labels $\ap_1 \sim_{\automaton} \ap_2$ do not essentially differ in $\automaton$, and thus one can simply replace $\ap_1$ with $\ap_2$ in $\MC$.

\begin{example}
	For the DRA in \cref{fig:motivatingDRA}, we have the compatibility relations  $\{a\}\{b\}\sim_{\automaton} \{b\}$ and $\{a\}\{a\}\sim_{\automaton} \{a\}$.
\end{example}

With the compatibility relation, we define the \emph{replaceability of edges}.

\begin{definition}[replaceability of edge]
 \label{definition:replaceability}
	An edge $\edge = (\state, \ap, \state')\in \States \times \Sigma\times \States$ in $\MC$ is \emph{replaceable} if
	$\state\not = \state'$, and
	there is a set $T\subseteq \etrace{\edge}$ such that
	\begin{itemize}
		\item for any $\infpathvar\in \eotrace{\edge}$, there is $\pathvar\in T$ such that $\pathvar$ is a prefix of $\infpathvar$,
		\item for any $\pathvar\in T$, the trace $\oalpha{\pathvar}$ is acceptance-preserving  and there is $\ap\in \Sigma$ such that $\oalpha{\pathvar} \sim_{\automaton} \ap$,
		\item for any $\pathvar \in T$, any strict prefix $\pathvar'$ of $\pathvar$ such that $\pathvar' \neq \edge$, and any $\ap \in \Sigma$, we have $\oalpha{\pathvar'} \not\sim_{\automaton} \ap$,
		\item $\edge\not\in T$, and 
		\item for any $\pathvar\in T$, $\state'$ appears only once in $\pathvar$.
	\end{itemize}
\end{definition}
\begin{example}
	For the MC in \cref{fig:motivatingMC}, the edges from $\state_0$ to $\state_1$, $\state_1$ to $\state_4$, and $\state_5$ to $\state_1$ are all replaceable.
\end{example}

Note that if an edge $e$ is replaceable, then there exists a unique set $T$ satisfying the above condition.
To see this, suppose that both $T_1$ and $T_2$ satisfy the conditions, and let $\pathvar \in T_1$.
By the first condition, there exists $\pathvar' \in T_2$ such that either $\pathvar$ is a prefix of $\pathvar'$ or vice versa.
This implies that $\pathvar = \pathvar'$, since $\oalpha{\pathvar} \sim_{\automaton} \ap_1$ and $\oalpha{\pathvar'} \sim_{\automaton} \ap_2$ for some $\ap_1, \ap_2 \in \Sigma$.

We call this unique set $T$ \emph{minimum compatible covering suffixes (MCCS)}, and denote the MCCS of a replaceable edge $e$ by $\coverSuffixes{e}{\automaton}$.
We prepare several auxiliary lemmata.
\begin{lemma}
	\label{lem:prefix}
	For any $\pathvar_1$ and $\pathvar_2$ in $\coverSuffixes{\edge}{\automaton}$, if $\pathvar_1$ is a prefix of $\pathvar_2$, then $\pathvar_1 = \pathvar_2$.
	\qed{}
\end{lemma}

\begin{lemma}
	\label{lem:characMCCS}
	We have
	\[
		 \transition_{\MC}(\state_{\edge}, \ap_{\edge}, \state'_{\edge}) =\sum_{\pathvar\in  \coverSuffixes{\edge}{\automaton}}\prob\big(\cyl(\pathvar)\big),
	\]
	where $e = (\state_{\edge}, \ap_{\edge}, \state'_{\edge})$.
\end{lemma}
\begin{proof}
	The cylinder set $\cyl(\state_{\edge} \cdot \ap_{\edge} \cdot \state'_{\edge})$ is equal to the union of the cylinder sets of $\pathvar \in \coverSuffixes{\edge}{\automaton}$,
	\ie $\cyl(\state_{\edge} \cdot \ap_{\edge} \cdot \state'_{\edge}) = \bigcup_{\pathvar \in \coverSuffixes{\edge}{\automaton}}\cyl(\pathvar)$, since for any $\infpathvar\in \eotrace{e}$, there is $\pathvar\in \coverSuffixes{\edge}{\automaton}$ such that $\pathvar$ is a prefix of $\infpathvar$.
	Moreover, for any $\pathvar_1\not = \pathvar_2\in\coverSuffixes{\edge}{\automaton} $, the sets $\cyl(\pathvar_1)$ and   $\cyl(\pathvar_2)$ are disjoint since $\pathvar_1$ cannot be a prefix of $\pathvar_2$ (and vice versa) by~\cref{lem:prefix}.
	Therefore, we have
	\begin{align*}
		\transition_{\MC}(\state_{\edge}, \ap_{\edge}, \state'_{\edge}) &= \prob(\cyl(\state_{\edge} \cdot \ap_{\edge} \cdot \state'_{\edge}))\\
		&=  \sum_{\pathvar\in \coverSuffixes{\edge}{\automaton}}\prob(\cyl(\pathvar)).
	\end{align*}
\end{proof}

Given a replaceable edge $e = (\state_{\edge}, \ap_{\edge}, \state'_{\edge})$, a label $\ap\in \Sigma$, and a state $\state\in S$, we write
$\replaceC{T}{\ap, \state}\subseteq \coverSuffixes{\edge}{\automaton}$ for the set of paths in $\coverSuffixes{\edge}{\automaton}$ such that
$\oalpha{\pathvar} \sim_{\automaton} \ap$ and $\last(\pathvar) = \state$.
Notice that the sets $\replaceC{T}{\ap, \state}$ form a partition of $\coverSuffixes{\edge}{\automaton}$ since we assume that $\ap_1\sim_{\automaton} \ap_2$ implies $\ap_1 = \ap_2$ w.l.o.g.
We also remark that $\replaceC{T}{\ap, \state'_{\edge}} = \emptyset$ for any $\ap\in \Sigma$.

We then define the \emph{replaced MC} by replacing a chosen replaceable edge on $\MC$.
\begin{definition}[replaced MC]
 \label{definition:replacedMC}
	Given a replaceable edge $e = (\state_{\edge}, \ap_{\edge}, \state'_{\edge})$ of $\MC$ \wrt $\automaton$,
	the \emph{replaced MC} $\replacedMC{\MC}{e}$ is defined by   $ \replacedMC{\MC}{e}\coloneqq (\States, \AP, \transition, s_{\mathrm{init}})$, where $\delta$ is defined as follows:
	\begin{align*}
		&\transition(\state, \ap, \state') \coloneqq \begin{cases}
			0 \qquad \qquad \qquad \qquad \qquad \quad \text{ if } (\state,\ap,\state') = e,\\
			\transition_{\MC}(\state, \ap,\state') \qquad \qquad \qquad \quad\text{ if } \state \not =  \state_{\edge},\\
			\transition_{\MC}(\state_{\edge}, \ap, \state') + \sum \limits_{\pathvar\in \replaceC{T}{\ap,\state'}}\prob(\pathvar) \text{ otherwise}.\\
		\end{cases}
	\end{align*}

\end{definition}

\begin{lemma}
	The construction of $\replacedMC{\MC}{e}$ is well-defined, that is, the function $\transition$ forms a transition probability.
\end{lemma}
\begin{proof}
	It suffices to show that $\sum_{(\ap, \state')} \transition(\state_{\edge}, \ap, \state') = 1$.
	We have
	\begin{align*}
		&\sum_{(\ap, \state')}\transition(\state_{\edge}, \ap, \state')\\
		= &\, 0 +  \sum_{(\ap, \state')\not = (\ap_{\edge}, \state'_{\edge})}\transition_{\MC}(\state_{\edge}, \ap, \state') + \sum_{\pathvar\in  \replaceC{T}{\ap,\state'}} \prob(\pathvar)\\
	    = & \sum_{\pathvar\in \coverSuffixes{\edge}{\automaton}}\prob(\pathvar) + \sum_{(\ap, \state')\not = (\ap_{\edge}, \state'_{\edge})}\transition_{\MC}(\state_{\edge}, \ap, \state') \\
		= &\sum_{(\ap, \state')}\transition_{\MC}(\state_{\edge}, \ap, \state') = 1,
	\end{align*}
	where we use~\cref{lem:characMCCS} to derive the last equality.
\end{proof}

\subsection{Correctness of Replaced MCs}
The following theorem ensures the correctness of the replaced MC.
\begin{theorem}[correctness of replacing]\label{theorem:correctness-ofreplacing}
	Given a replaceable edge $\edge$ of $\MC$ \wrt $\automaton$,
	we have
	\[
	\prob(\MC\vDash \automaton) = \prob(\replacedMC{\MC}{\edge}\vDash \automaton).\]
\end{theorem}
For the rest of this section, we provide our proof of~\cref{theorem:correctness-ofreplacing}.

Towards this,  we prepare two important lemmata.

\begin{lemma}
	\label{lem:setPath}
	Let $\edge = (\state_{\edge}, \ap_{\edge}, \state'_{\edge})$ be a replaceable edge of $\MC$ \wrt $\automaton$.

	For each edge $(\state, \ap, \state')$ of $\replacedMC{\MC}{\edge}$, there is a set $T_{(\state, \ap, \state')}$ of paths of $\MC$ such that
	\begin{itemize}
		\item $\transition_{\replacedMC{\MC}{\edge}}(\state, \ap, \state')  = \sum_{\pathvar \in T_{(\state, \ap, \state')}}\prob_{\MC}(\pathvar)$,
		\item $\oalpha{\pathvar} \sim_{\automaton} \ap$, $\state = \initpath(\pathvar)$, and $\state' = \last(\pathvar)$ for any path $\pathvar \in T_{(\state, \ap, \state')}$.
	\end{itemize}

\end{lemma}
\begin{proof}
    For each edge $(\state, \ap, \state')$ of $\replacedMC{\MC}{\edge}$,
	we first define $X_{(\state, \ap, \state')}$ as follows:
	\begin{align*}
		X_{(\state, \ap, \state')}\coloneqq \begin{cases}
			\big\{(\state, \ap, \state') \big\} &\text{ if $\state\not = \state_{\edge}$, }\\
			\big\{(\state, \ap, \state') \big\} \cup\replaceC{T}{\ap, \state'}  &\text{ otherwise. }
		\end{cases}
	\end{align*}
	Note that $(\state, \ap, \state')$ cannot be $\edge$.
	We then define $T_{(\state, \ap, \state')}$ as follows:
	\begin{align*}
		T_{(\state, \ap, \state')}\coloneqq \begin{cases}
			X_{(\state, \ap, \state')} \backslash \{(\state,\ap,\state')\} &\text{ if $(\state,\ap,\state')\not \in \Edges_{\MC}$, }\\
			X_{(\state, \ap, \state')} &\text{ otherwise.}
		\end{cases}
	\end{align*}
	It is straightforward to check that the family $\big(T_{(\state, \ap, \state')}\big)_{(\state, \ap, \state')}$ satisfies the above condition by the definitions of the replaceability of edges and replaced MC.
\end{proof}
Clearly, two distinct sets $T_{(\state_1, \ap_1, \state'_1)}$ and $T_{(\state_2, \ap_2, \state'_2)}$ are disjoint.

\begin{lemma}
	\label{lem:BSCCpreserve}
	Let $\edge = (\state_{\edge}, \ap_{\edge}, \state'_{\edge})$ be a replaceable edge and $X\subseteq \States\times Q$ be a BSCC in $\replacedMC{\MC}{\edge} \otimes \automaton$.
	There is a (unique) BSCC $Y$ in $\MC\otimes \automaton$ such that $X\subseteq Y$, and moreover, $X$ is accepting  \revision{if and only if} $Y$ is accepting.
\end{lemma}
\begin{proof}
	We can immediately see that $X$ is an SCC in $\MC\otimes \automaton$.
	Suppose that there is a state $(\state, \position)\not \in X$ such that
	the state $(\state, \position)$ is reachable from $X$ in $\MC\otimes \automaton$;
	if there are no such states, we can conclude that $X$ itself is a BSCC in $\MC\otimes \automaton$.
	We show that
	there is a path from  such $(\state, \position)$ to $X$ in $\MC\otimes \automaton$, which implies that $X$ is included in the BSCC
	$Y\coloneqq \{(\state, \position)\mid (\state, \position)\text{ is reachable from $X$ in $\MC\otimes \automaton$}\}$.

	Take a finite path  $\pathvar$ from $(\state_0, \position_0)\in X$ to  $(\state, \position)$ in $\MC\otimes \automaton$;
	without loss of generality, we assume that $\pathvar$ does not have any cycles,
	and $(\state_0, \position_0)$ is the unique state that is contained in $X$ in $\pathvar$.
	Under this assumption, the prefix $\pathvar_0 = \big((\state_0, \position_0), \ap, (\state_{1}, \position_{1})\big)$ of $\pathvar$ is an edge that does not exist in $\replacedMC{\MC}{\edge} \otimes \automaton$.
	By the construction of $\replacedMC{\MC}{\edge} \otimes \automaton$, we have $\state_0 = \state_{\edge}$, $\ap = \ap_{\edge}$, $\state_1 = \state'_{\edge}$, and $\position_1 = \DRATransition_{\automaton}(\position_0, \ap_{\edge})$.
	That is, the prefix $\pathvar_0$ is the product edge induced by the replaceable edge $\edge$.
	Then, by the definition of the replaceability of edges, we can see that there is a state $(\state_x, \position_x)\in X$ that satisfies one of the following conditions:
	\begin{itemize}
		\item the state $(\state_x, \position_x)$ appears in the suffix $\pathvar_{1}$ starting from $(\state_{1}, \position_{1})$,
		\item the state $(\state_x, \position_x)$ is reachable from $(\state, \position)$ in $\MC\otimes \automaton$,
	\end{itemize}
	where $\pathvar_{1}$ is the suffix of $\pathvar$ that satisfies $\pathvar = \pathvar_0\cdot \pathvar_1$.
	However, the first case does not hold because it contradicts the assumption of the uniqueness of $(\state_0, \position_0)$ in $\pathvar$.
	We thus see that there is a path from $(\state, \position)$ to $X$.

\end{proof}

\begin{example}
	In~\cref{lem:BSCCpreserve}, we cannot expect $X = Y$ to hold in general.
Consider the left MC in~\cref{fig:counterEGMC}.
With respect to the DRA in~\cref{fig:motivatingDRA}, the edge from $\state_0$ to $\state_1$ is replaceable, and the right MC in~\cref{fig:counterEGMC} is the replaced MC obtained by the replacement of the edge.
\cref{fig:counterEGProd} illustrates the product MCs of these two MCs.
Clearly, the unique BSCC in the left product MC is strictly larger than that in the right product MC.
\end{example}

\begin{figure}[t]
	\subfloat[t][The original MC.]{%
	\begin{minipage}[t]{0.45\columnwidth}
	\begin{tikzpicture}[
		->,                    %
		>=Stealth,             %
		shorten >=1pt,         %
		auto,                  %
		node distance=2.5cm,   %
		semithick,              %
		every state/.style={
		font=\small,
		minimum size=6mm
		}
		]
		\node[state, initial] (q0) at (0, 0) {$\state_0$};
		\node[state] (q1) at (2, 0) {$\state_1$};

		\path
		(q0) edge[bend left] node {$\{a\}, 1$} (q1)
		(q1) edge[bend left] node {$\{b\}, 1$} (q0);
	\end{tikzpicture}
	\end{minipage}
	}
	\hfill
	\subfloat[t][The MC after replacing the edge from $\state_0$ to $\state_1$.]{%
	\begin{minipage}[t]{0.45\columnwidth}
	\begin{tikzpicture}[
	->,                    %
	>=Stealth,             %
	shorten >=1pt,         %
	auto,                  %
	node distance=2.5cm,   %
	semithick,              %
	every state/.style={
		font=\small,
		minimum size=6mm
	}
	]
	\node[state, initial] (q0) at (0, 0) {$\state_0$};
	\node[state] (q1) at (2, 0) {$\state_1$};

	\path
		(q1) edge[bend left] node {$\{b\}, 1$} (q0)
    (q0)[loop above] edge node {$\{b\}, 1$} ();
	\end{tikzpicture}
	\end{minipage}
	}
\caption{
	An MC and its replaced MC.
}
	\label{fig:counterEGMC}
\end{figure}
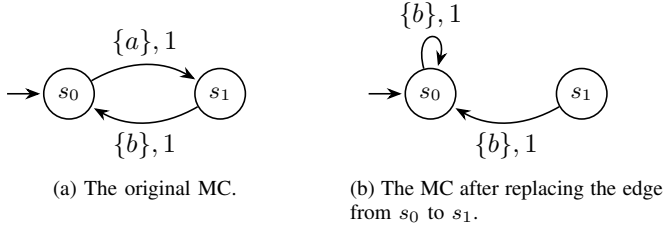

\begin{figure}[t]
    \subfloat[t][The product MC of the original MC.]{
	\begin{minipage}[t]{0.45\columnwidth}
	\begin{tikzpicture}[
		->,                    %
		>=Stealth,             %
		shorten >=1pt,         %
		auto,                  %
		node distance=2.5cm,   %
		semithick,              %
		every state/.style={
		font=\small,
		minimum size=6mm
		},scale=0.90,every node/.style={initial text={},transform shape}
		]
		\node[state, initial, accepting] (q0) at (0, 0) {$\state_0, \position_0$};
		\node[state] (q1) at (2, 0) {$\state_1, \position_1$};
		\node[state, accepting] (q2) at (0, -2) {$\state_1, \position_0$};
		\node[state] (q3) at (2, -2) {$\state_0, \position_1$};

		\path
		(q0) edge[bend left] node {$\{a\},1$} (q1)
		(q2) edge node {$\{b\},1$} (q0)
		(q3) edge node[right] {$\{a\},1$} (q1)
		(q1) edge[bend left] node {$\{b\}, 1$} (q0);
	\end{tikzpicture}
	\end{minipage}
	}
	\hfill
	\subfloat[t][The product MC of the replaced MC.]{%
	\begin{minipage}[t]{0.45\columnwidth}
	\begin{tikzpicture}[
		->,                    %
		>=Stealth,             %
		shorten >=1pt,         %
		auto,                  %
		node distance=2.5cm,   %
		semithick,              %
		every state/.style={
		font=\small,
		minimum size=6mm
		},scale=0.90,every node/.style={initial text={},transform shape}
		]
		\node[state, initial, accepting] (q0) at (0, 0) {$\state_0, \position_0$};
		\node[state] (q1) at (2, 0) {$\state_1, \position_1$};
		\node[state, accepting] (q2) at (0, -2) {$\state_1, \position_0$};
		\node[state] (q3) at (2, -2) {$\state_0, \position_1$};

		\path
		(q2) edge node {$\{b\}, 1$} (q0)
		(q3) edge node[pos=0.3,above right] {$\{b\}, 1$} (q0)
		(q1) edge node[above] {$\{b\}, 1$} (q0)
		(q0)[loop above] edge node {$\{b\}, 1$} ();
	\end{tikzpicture}
	\end{minipage}
	}
\caption{
	The product MCs of~\cref{fig:counterEGMC} with the automaton~\cref{fig:motivatingDRA}.
}
	\label{fig:counterEGProd}
\end{figure}
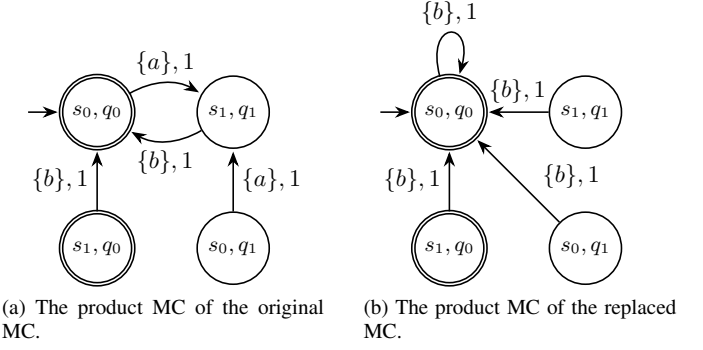

We conclude this section by proving the theorem with the help of the above two lemmata.
\begin{proof}[Proof of \cref{theorem:correctness-ofreplacing}]
	Let $X$ be a BSCC in $\replacedMC{\MC}{\edge} \otimes \automaton$.
	By~\cref{lem:BSCCpreserve}, we have the BSCC $Y$ that includes $X$ in $\MC\otimes \automaton$.
	It suffices to show the following inequality:
	\begin{gather}
		\label{eq:keyEQ}
		 \rpr_{\replacedMC{\MC}{\edge}\otimes \automaton}(X) \leq \rpr_{\MC\otimes \automaton}(Y).
	\end{gather}
	To see this, let $\widehat{X}$ be the union $\cup_j X_j $ of the accepting BSCCs $X_j$ in $\replacedMC{\MC}{\edge} \otimes \automaton$,
	 and let $\widehat{Y}$ be the union $\cup_k Y_k $ of the accepting BSCCs $Y_k$ in $\MC \otimes \automaton$.
	Similarly, let $\widetilde{X}$ and  $\widetilde{Y}$ be the union of non-accepting BSCCs in $\replacedMC{\MC}{\edge} \otimes \automaton$ and $\MC \otimes \automaton$, respectively.
	We have the following inequalities:
	\begin{align*}
		\rpr_{\replacedMC{\MC}{\edge}\otimes \automaton}(\widehat{X}) &= \sum_{j} \rpr_{\replacedMC{\MC}{\edge}\otimes \automaton}(X_j)\\
		 &\leq \sum_{k} \rpr_{\MC\otimes \automaton}(Y_k)
		= \rpr_{\MC\otimes \automaton}(\widehat{Y}),\\
		\rpr_{\replacedMC{\MC}{\edge}\otimes \automaton}(\widehat{X}) &= 1 - \rpr_{\replacedMC{\MC}{\edge}\otimes \automaton}(\widetilde{X}) \\
		&\geq 1 - \rpr_{\MC\otimes \automaton}(\widetilde{Y}) = \rpr_{\MC\otimes \automaton}(\widehat{Y}).
	\end{align*}
	The second inequality holds due to the almost-sure reachability of BSCCs (see~\cref{lem:asBSCC}).
	The above two inequalities imply that $\rpr_{\replacedMC{\MC}{\edge}\otimes \automaton}(\widehat{X}) = \rpr_{\MC\otimes \automaton}(\widehat{Y}) $.

	We conclude the proof by showing~\cref{eq:keyEQ}.
	For a finite path $\pathvar$ in $\MC$ and a position $\position$ in $\automaton$, we write $\pathvar(\position)$ for the unique path in $\MC\otimes \automaton$ from $(\initpath(\pathvar), \position)$ induced by $\pathvar$.
	This is well-defined because the product construction keeps the original labels on transitions.

	Take a path
	\[
		\pathvar = \big((\state_0, \position_0), \ap_1, (\state_1, \position_1)\big)\cdots \big((\state_{m-1}, \position_{m-1}), \ap_m, (\state_m, \position_m)\big)
	\]
	from $(\initialState, \initialPosition) = (\state_0, \position_0)$ to $(\state_m, \position_m) \in X$ in $\replacedMC{\MC}{\edge}\otimes \automaton$.
	Let $\ddot{\pathvar} \coloneqq \edge_1\cdots \edge_m$ be the unique path in $\replacedMC{\MC}{\edge}$ that induces $\pathvar$, where $\edge_i \coloneqq (\state_{i-1}, \ap_i, \state_i)$ for each $i\in \{1,\dots,m\}$.
	We construct a set $P(\ddot{\pathvar})$ of paths from $(\initialState, \initialPosition)$ to $X(\subseteq Y)$ in $\MC\otimes \automaton$.
	Specifically, we let $P(\ddot{\pathvar}) \coloneqq \{\pathvar_1(\initialPosition)\cdot \pathvar_2(\position_1)\cdots \pathvar_m(\position_{m-1}) \mid \pathvar_i\in T_{\edge_i} \text{ for all }i=1,\dots, m\}$, 
	where $\position_i$ is the last position in $\pathvar_i(\position_{i-1})$.
	Note that these positions $\position_i$ do not depend on the choice of $\pathvar_i\in T_{\edge_i}$.
	This is because for any $\pathvar_i\in T_{\edge_i}$, we have $\oalpha{\pathvar_i}\sim_{\automaton} \ap_j$ by~\cref{lem:setPath}, and thus, $\pathvar_i(\position_{i-1})$ ends in $(\state_i, \position_i)$.

	Moreover, distinct paths $\pathvar_i$ in $T_{\edge_i}$ from  $i= 1$ to $i = m$ yield the paths in $P(\ddot{\pathvar})$ whose cylinder sets are disjoint. 
	To see this, 
	let $\pathvar_i$ and $ \pathvar'_i$ in $T_{\edge_i}$ from $i = 1$ to $i = m$ and consider the two paths $\pathvar_1\cdots\pathvar_m$ and $\pathvar'_1\cdots\pathvar'_m$. 
	Assume that there is $j$ such that $\pathvar_j\not = \pathvar'_j$; take the least $j$ among them.  
	Then these two segments are distinct members of $T_{\edge_j}$.
	By the construction of $T_{\edge_j}$ in~\cref{lem:setPath}, distinct paths in $T_{\edge_j}$ have disjoint cylinder sets.
	Thus, the cylinder sets of $\pathvar_1\cdots\pathvar_m$ and $\pathvar'_1\cdots\pathvar'_m$ are disjoint.

	By the construction, we have
	\begin{align*}
		\prob_{\MC\otimes \automaton}\big(P(\ddot{\pathvar})\big)
		&= \sum_{\pathvar_1\in T_{\edge_1},\dots,\pathvar_m\in T_{\edge_m}} \prod^{m}_{i=1}\prob_{\MC}(\pathvar_i)\\
		&= \prod^{m}_{i=1}\sum_{\pathvar_i\in T_{\edge_i}}\prob_{\MC}(\pathvar_i)\\
		&= \prod^{m}_{i=1}\transition_{\replacedMC{\MC}{\edge}}(\edge_i)\\
		&= \prob_{\replacedMC{\MC}{\edge}\otimes \automaton}(\pathvar),
	\end{align*}
	where the third equality follows from~\cref{lem:setPath}.

	Let $Z$ be the set of paths from $(\initialState, \initialPosition)$ to $X$ in $\replacedMC{\MC}{\edge}\otimes \automaton$.
	For distinct $\pathvar_1, \pathvar_2 \in Z$, $\ddot{\pathvar_1}$ and $\ddot{\pathvar_2}$ are distinct, and thus, $P(\ddot{\pathvar_1})$ and $P(\ddot{\pathvar_2})$ are also disjoint.
	Therefore, we have
	\begin{align*}
		\rpr_{\replacedMC{\MC}{\edge}\otimes \automaton}(X)
		&= \sum_{\pathvar\in Z}\prob_{\replacedMC{\MC}{\edge}\otimes \automaton}(\pathvar)\\
		&= \sum_{\pathvar\in Z}\prob_{\MC\otimes \automaton}\big(P(\ddot{\pathvar})\big)\\
		&= \prob_{\MC\otimes \automaton}\Big(\bigcup_{\pathvar\in Z}P(\ddot{\pathvar})\Big)\\
		&\leq \rpr_{\MC\otimes \automaton}(Y).
	\end{align*}
	This implies~\cref{eq:keyEQ}.

\end{proof}

\section{Erasability of States}
\label{sec:eraseStates}

We continue to fix $\MC$ as an MC and $\automaton$ as a DRA.
We introduce the notion of \emph{erasability of states}, which forms the foundation of our specification-guided \revision{path shortcutting}.
\revision{In a nutshell, we eliminate a state if all its incoming edges are replaceable. In fact, our algorithm, which is formally introduced in~\cref{section:algorithm}, attempts to eliminate each state one by one by checking whether all its incoming edges are replaceable.}

\begin{definition}[erasability of state]
    \label{definition:erasability}
	A state $\state$ is \emph{erasable} if $\state$ is not the initial state and all incoming edges $\InEdges{\state}$ of $\state$ are replaceable.
\end{definition}

\begin{example}
	In the MC shown in~\cref{fig:motivatingMC}, the states $\state_1$ and $\state_4$ are erasable, while no other states are erasable.
\end{example}

As expected, given an erasable state $\state$,  we can construct an equivalent MC $\erase{\MC}{\state}$  that does not contain the erasable state $\state$ from $\MC$.
\begin{definition}[erased MC]
	Let $\state$ be an erasable state.
	The \emph{erased MC} $\erase{\MC}{\state}$ for $\MC$ and $\state$ is given by
	the MC $\replacedMC{\MC}{\InEdges{\state}}$ that excludes the state $\state$,
	where  $\replacedMC{\MC}{\InEdges{\state}}$ is obtained by iteratively applying the edge replacement of \cref{definition:replacedMC} to each $\edge \in \InEdges{\state}$.
\end{definition}

\begin{lemma}
	\label{lem:wd-erased}
	The \emph{erased MC} $\erase{\MC}{\state}$ is well-defined.
\end{lemma}
\begin{proof}
	Suppose that there are two distinct replaceable edges $\edge_1$ and $\edge_2$ that are in $\InEdges{\state}$.
	We first see that the edge $\edge_2$ is still replaceable in $\replacedMC{\MC}{\edge_1}$.
	This is in fact trivial since the MCCS  $\coverSuffixes{\edge_2}{\automaton}$ does not change between $\MC$ and $\replacedMC{\MC}{\edge_1}$.
	It is also straightforward to see that  $\replacedMC{\replacedMC{\MC}{\edge_2}}{\edge_1} = \replacedMC{\replacedMC{\MC}{\edge_1}}{\edge_2}$.
\end{proof}

We finally present our main theorem: the correctness of erased MCs.

\begin{theorem}[correctness of erasing]
    \label{theorem:correctness_of_erasing}
	Given an erasable state $\state$ of $\MC$ \wrt $\automaton$,
	we have
	\[
		\prob(\MC\vDash \automaton) = \prob(\erase{\MC}{\state}\vDash \automaton).
	\]
\end{theorem}

\begin{proof}
	This is an immediate consequence of~\cref{lem:wd-erased} and~\cref{theorem:correctness-ofreplacing}.
\end{proof}

We conclude this section by presenting an observation on the erasability of states.

\begin{proposition}
	\label{proposition:erasability_single}
	Suppose that for any $\ap_1$ and $\ap_2\in \Sigma$, the word $\ap_1\cdot \ap_2$ is acceptance-preserving and there is $\ap\in \Sigma$ such that $\ap_1\cdot \ap_2 \sim_{\automaton}\ap$.

	If there are no erasable states in $\MC$ \wrt $\automaton$, then every state $\state$ except the initial state in $\MC$ has a self-loop.
\end{proposition}
\begin{proof}
	Assume that there is a state $\state$ that is not the initial state and does not have any self-loops.
	Since there are no erasable states in $\MC$, there is an incoming edge $\edge\in \InEdges{\state}$.
	This immediately leads to a contradiction since such incoming edges are all replaceable due to the assumption.
\end{proof}

The assumption in \cref{proposition:erasability_single} holds, for instance, for any LTL formula of the form $\Eventually \psi$, where $\psi$ is a Boolean combination of atomic propositions.
One concrete example is the abstraction of the MC shown in \cref{figure:mc_haddad} to \cref{figure:product_mc_haddad}.

\section{Specification-guided \revision{Path Shortcutting}}\label{section:algorithm}

\begin{algorithm}[tbp]
 \ShortVersion{\footnotesize}\LongVersion{\small}
 \caption{Specification-guided \revision{path shortcutting} of MCs.}%
 \label{algorithm:abstraction}
 \DontPrintSemicolon{}
 \SetKwFunction{FConstructMCCS}{constructMCCS}
 \SetKwFunction{FErase}{Erase}

 \Input{An MC $\MC = (\States, \AP, \MCTransition, \initialState)$, a DRA $\automaton$, $\MCCSBound \in \setN$}
 \Output{An MC $\reducedMC$ satisfying $\prob(\MC \vDash \automaton) = \prob(\reducedMC \vDash \automaton)$}

 $\reducedMC \gets \MC$\;
 \For{$\state \in \States \setminus \{ \initialState \}$}{
   \If{$\exists \ap \in \Sigma.\, (\state, \ap, \state) \in \InEdges{\state}$} {
     \Continue{}\;
   }
   \If{$\exists \edge \in \InEdges{\state}.\, \FConstructMCCS(\revision{\MC'}, \edge, \automaton, \MCCSBound) = \bot$ \nllabel{algorithm:abstraction:check_erasable}} {
     \Continue{}
     \tcp*[r]{MCCS construction failed}
   }
   $\reducedMC \gets \erase{\reducedMC}{\state}$\;\nllabel{algorithm:abstraction:erase}
 }
 \Return{} $\reducedMC$\;
\end{algorithm}

\begin{algorithm}[tbp]
  \ShortVersion{\footnotesize}\LongVersion{\small}
  \caption{Bounded MCCS construction.}%
  \label{algorithm:MCCS_construction}
  \DontPrintSemicolon{}
  \SetKwFunction{FConstructMCCS}{constructMCCS}
  \Fn{\FConstructMCCS{$\MC, \edge, \automaton, \MCCSBound$}}{
    \Input{An MC $\MC$, an edge $\edge = (\state, \ap, \state')$ of $\MC$, a DRA $\automaton$, and a bound $\MCCSBound \in \setN$}
    \Output{The MCCS $\MCCS$ of $\edge$ if it is found; otherwise $\bot$}

    $\MCCSCandidate \gets \{ \edge \}$;\,\,$\MCCS \gets \emptyset$\;
    \While{$\MCCSCandidate \neq \emptyset$}{
      \KwPop{} $\pathvar$ \KwFrom{} $\MCCSCandidate$\;
      $\tilde{\state} \gets \last(\pathvar)$\;
      \For{$\tilde{\edge} = (\tilde{\state}, \tilde{\ap}, \tilde{\state}') \in \OutEdges{\tilde\state}$}{
        \If{$\tilde{\state}' = \state'$} {
          \Return{} $\bot$\;
          \nllabel{algorithm:MCCS_construction:cycle}
        } \ElseIf{$|\pathvar| > \MCCSBound$} {
          \Return{} $\bot$ \;
          \nllabel{algorithm:MCCS_construction:too_long}
        } \ElseIf{$\exists \ap' \in \Sigma.\, \oalpha{\pathvar \cdot \tilde{\edge}} \sim_{\automaton} \ap'$} {
          \KwPush{} $\pathvar \cdot \tilde{\edge}$ \KwTo{} $\MCCS$\;
          \nllabel{algorithm:MCCS_construction:compatible}
        } \Else {
          \KwPush{} $\pathvar \cdot \tilde{\edge}$ \KwTo{} $\MCCSCandidate$\;
          \nllabel{algorithm:MCCS_construction:extend}
        }
      }
    }
    \Return{} $\MCCS$\;
  }
\end{algorithm}

\cref{algorithm:abstraction} outlines our algorithm for abstracting an MC $\MC$ with respect to a DRA $\automaton$, where the definition of $\FConstructMCCS$ is given in \cref{algorithm:MCCS_construction}.
In \cref{algorithm:abstraction}, we try to remove each non-initial state $\state$ of $\MC$.
We perform this removal trial only once for each state throughout the entire loop due to the following monotonicity in replaceability.

\begin{proposition}
 Let $\MC$ be an MC, and let $\edge$ and $\edge'$ be distinct edges of $\MC$.
 Suppose $\edge$ is replaceable.
 The edge $\edge'$ is replaceable in $\replacedMC{\MC}{\edge}$ only if it is replaceable in $\MC$.
\end{proposition}
\begin{proof}
 Suppose $\edge'$ is replaceable in $\replacedMC{\MC}{\edge}$ and let $\MCCS$ be the MCCS of $\edge'$ in $\replacedMC{\MC}{\edge}$.
 If $\MCCS$ does not contain edges introduced in replacing $\edge$, $\MCCS$ is also the MCCS of $\edge'$, and thus, $\edge'$ is replaceable in $\MC$.

 Otherwise, let $\MCCSCandidate$ be the set of paths obtained by replacing each edge in $\MCCS$ introduced by replacing $\edge$ with the corresponding paths in $\MC$.
 Namely, each edge $(\state, \ap, \state')$ of $\replacedMC{\MC}{\edge}$ in $\MCCS$ is replaced with a set $T_{(\state, \ap, \state')}$ of paths of $\MC$ defined in \cref{lem:setPath}.
 From the definition of $\replacedMC{\MC}{\edge}$,
 $\MCCSCandidate$ satisfies all the conditions in \cref{definition:replaceability} other than the last condition.
 Let $\MCCS'$ be the set of (non-strict) prefixes $\pathvar$ of paths in $\MCCSCandidate$ such that
 \begin{ienumeration}
   \item $|\pathvar| > 1$,
   \item $\oalpha{\pathvar}$ is acceptance-preserving,
   \item there is $\ap \in \Sigma$ such that $\oalpha{\pathvar} \sim_{\automaton} \ap$, and
   \item for any strict prefix $\pathvar'$ of $\pathvar$, $\oalpha{\pathvar'}$ is not acceptance-preserving or there is no $\ap' \in \Sigma$ satisfying $\oalpha{\pathvar'} \sim_{\automaton} \ap'$.
 \end{ienumeration}
 For each $\pathvar \in \MCCSCandidate$, exactly one prefix of $\pathvar$ is in $\MCCS'$, and $\MCCS'$ is the MCCS of $\edge'$ in $\MC$.
 Therefore, $\edge'$ is replaceable in $\MC$.
\end{proof}

For each $\state \in \States \setminus \{\initialState\}$,
we check if $\state$ is erasable by trying to construct the MCCS for each incoming edge of $\state$ (\cref{algorithm:abstraction:check_erasable} of \cref{algorithm:abstraction}).
We use the bounded MCCS construction in \cref{algorithm:MCCS_construction} for practical efficiency.
Namely, we try to construct the MCCS only using the paths of length at most $\MCCSBound$, and if the construction fails, we deem $\edge$ not replaceable.

In \cref{algorithm:MCCS_construction}, we maintain a set $\MCCSCandidate$ of candidate paths $\pathvar$ and gradually extend them by appending an edge $\tilde\edge$.
For each such extension, we first check whether the target state $\state'$ of the root edge $\edge$ reappears in the path; if it does, we conclude that $\edge$ is not replaceable and return $\bot$ (\cref{algorithm:MCCS_construction:cycle}).
Otherwise, if the current candidate path $\pathvar$ already exceeds the bound $\MCCSBound$, we conclude that the bounded MCCS construction fails and return $\bot$ (\cref{algorithm:MCCS_construction:too_long}).
If neither of the above cases applies and $\oalpha{\pathvar \cdot \tilde\edge}$ is compatible with some $\ap \in \Sigma$, we add $\pathvar \cdot \tilde\edge$ to $\MCCS$ (\cref{algorithm:MCCS_construction:compatible}).
Otherwise, we add $\pathvar \cdot \tilde\edge$ to $\MCCSCandidate$ for further exploration (\cref{algorithm:MCCS_construction:extend}).
We return $\MCCS$ if $\MCCSCandidate$ becomes empty.
Such $\MCCS$ is the MCCS of $\edge$ intuitively because we cover all the infinite paths starting with $\edge$ and truncated at the first acceptance-preserving edge.

If we find such $\MCCS$ for each incoming edge of $\state$, $\state$ is erasable (\cref{definition:erasability}), and we remove $\state$ from $\reducedMC$ (\cref{algorithm:abstraction:erase} of \cref{algorithm:abstraction}).
\cref{theorem:correctness_of_erasing} guarantees that such removal does not change the satisfaction probability of $\automaton$.
We let $\reducedMCFinal$ be the abstraction of an MC $\MC$ with respect to a DRA $\automaton$.

\section{Experimental Evaluation}\label{section:experiments}

We implemented \revision{a prototype tool, \ourTool{},} in C++ using Spot~\cite{Duret-LutzRCRAS22} version \revision{2.15.1}\footnote{\ourTool{} is publicly available on \url{https://github.com/SoftwareFoundationGroupAtKyotoU/specification-guided-path-shortcutting}.}.
	Given a state-labeled MC $\tilde\MC$ and an LTL formula $\formula$,
	\revision{\ourTool{}}
	\begin{ienumeration}
	    \item translates $\tilde\MC$ into a transition-labeled MC $\MC$;
	    \item constructs a DRA $\automaton_{\formula}$ from $\formula$;
		\item abstracts $\MC$ into $\reducedMCFinal[\automaton_{\formula}]$ with respect to $\automaton_{\formula}$ using \cref{algorithm:abstraction};
		\item eliminates the labels on the edges of the product MC $\reducedMCFinal[\automaton_{\formula}] \otimes \automaton_{\formula}$ to obtain an unlabeled MC; and
	    \item invokes \storm{}~\cite{HenselJKQV22} on the product MC together with an LTL formula $\formula'$ encoding the acceptance condition of $\automaton_{\formula}$.
	\end{ienumeration}
	Concretely, if an acceptance condition is $\AccCondition = (L_i, U_i)_{i \in I}$, we use $\formula' \coloneqq \bigvee_{i \in I} \bigl((\Globally \Eventually L_i)  \land (\Eventually \Globally \lnot U_i)\bigr)$, where each state in this DRN encoding is labeled according to whether its DRA component belongs to $L_i$ and $U_i$.

We conducted experiments to answer the following research questions.
\begin{description}
	\item[RQ1.] \revision{Does \ourTool{} outperform \storm{} in terms of the efficiency of probabilistic model checking?}
    \item[\revision{RQ2.}] \revision{What is the isolated contribution of specification-guided path shortcutting, compared with the same \ourTool{} workflow without shortcutting?}
    \item[\revision{RQ3.}] \revision{Does the combination of specification-guided \revision{path shortcutting} with bisimulation minimization~\cite{KatoenKZJ07} further improve the performance of probabilistic model checking?}
	\item[\revision{RQ4.}] How sensitive is specification-guided \revision{path shortcutting} to the MCCS bound $\MCCSBound$?
\end{description}

\subsection{Benchmarks}\label{section:experimens:benchmarks}

\begin{table}[tbp]
 \centering
 \caption{Summary of the MC in each benchmark. The ``\# of states'' column shows the number of states before abstraction.}%
 \label{table:benchmark_summary}
 \LongVersion{\small}\ShortVersion{\footnotesize}
 \begin{tabular}{lrr}
  \toprule
  Benchmark & \# of states & Parameters in the model\\
  \midrule
  \brp{} & 5,192 & $N=64, \mathit{MAX}=5$ \\
  \crowds{} & 359,622 & $\mathit{TotalRuns}=6, \mathit{CrowdSize}=10$ \\
  \egl{} & 115,710 & $N=5,L=6$\\
  \leader{} & 4,244 & $N=5, K=4$\\
  \nand{} & 18,826,082 & $N=60,K=4$\\
  \haddad{} & 101 & $N=50,p=0.7$ \\
  \bottomrule
 \end{tabular}
\end{table}

For the evaluation, we used six benchmarks: \brp{}, \crowds{}, \egl{}, \leader{}, \nand{}, and \haddad{}.
Each benchmark consists of an MC and multiple LTL formulas used as verified properties.
\cref{table:benchmark_summary} summarizes the MCs and \cref{table:benchmark_properties} summarizes the LTL formulas for each benchmark.
We used the PRISM files available from \url{https://qcomp.org}~\cite{HartmannsKPQR19}.
Most of the LTL formulas are our original ones, designed to be more complex than the properties in~\cite{HartmannsKPQR19}.

\begin{table*}[tbp]
 \centering
 \caption{LTL formulas in our benchmarks and the corresponding results of the experiments conducted to answer RQ1. The ``Total Time'' columns show the total time taken for the entire workflow. The ``Reduced States'' column shows the number of states eliminated by our specification-guided \revision{path shortcutting}. \revision{The ``Trans. Diff.'' column shows the difference between the number of transitions of the product MCs with and without specification-guided path shortcutting, where a negative value indicates that our abstraction reduced the number of transitions.} The ``Prep.\ Time'' and ``Storm Time'' columns show the times taken for the preprocessing \revision{(\eg{} abstraction and product construction)} and model checking with \storm{}, respectively. All time columns are measured in seconds and report the mean over 30 runs. The better total time between the two approaches is highlighted.}%
 \label{table:benchmark_properties}
 \label{table:properties_brp}
 \label{table:properties_crowds}
 \label{table:properties_egl}
 \label{table:properties_haddad}
 \label{table:properties_leader}
 \label{table:properties_nand}
 \footnotesize
 \revision{\begin{tabular}{llr|rrrrrrr}
  \toprule
  & \multirow{3}{*}{LTL formula} & \multicolumn{1}{c|}{\storm{}} & \multicolumn{5}{c}{Ours} \\
  & & Total & Total & Reduced & Trans.\ & Prep.\ & Storm \\
  & & Time  & Time  & States  & Diff.\  & Time  & Time \\
  \midrule
  $\formula^1_{\brp}$ & $\Eventually \mathtt{error}$ & 0.38 & \tbcolor{0.079} & 1723 & -1467 & 0.0099 & 0.070 \\
  $\formula^2_{\brp}$ & $\Eventually (\mathtt{error} \land \mathtt{uncertain})$ & 0.35 & \tbcolor{0.054} & 1723 & -1467 & 0.0096 & 0.046 \\
  $\formula^3_{\brp}$ & $\Eventually (\lnot \mathtt{idle} \land \lnot \mathtt{received})$ & 0.33 & \tbcolor{0.25} & 1723 & -1467 & 0.0092 & 0.25 \\
  $\formula^4_{\brp}$ & $\Globally \bigl(\mathtt{retransmit} \implies ((\mathtt{wait\_ack} \lor \mathtt{retransmit}) \Until \mathtt{success})\bigr)$ & 0.35 & \tbcolor{0.094} & 1083 & -1083 & 0.0091 & 0.087 \\
  $\formula^5_{\brp}$ & $\Eventually \bigl(\mathtt{next\_frame} \land \Next (\lnot \mathtt{retransmit} \Until \mathtt{next\_frame})\bigr)$ & 0.35 & \tbcolor{0.26} & 1723 & -2786 & 0.0097 & 0.25 \\
  $\formula^6_{\brp}$ & $\lnot \mathtt{wait\_ack} \Until (\mathtt{msg\_lost} \land \Eventually \mathtt{resync})$ & 0.39 & \tbcolor{0.077} & 1723 & -2920 & 0.011 & 0.068 \\
  $\formula^7_{\brp}$ & $\lnot \mathtt{msg\_lost} \Until (\mathtt{resync} \land \Eventually \mathtt{ack\_lost})$ & 0.36 & \tbcolor{0.049} & 1723 & -2545 & 0.010 & 0.040 \\
  $\formula^8_{\brp}$ & $\lnot \mathtt{resync} \Until (\mathtt{ack\_lost} \land \Eventually \mathtt{success})$ & 0.36 & \tbcolor{0.26} & 1723 & -2915 & 0.010 & 0.25 \\
  $\formula^9_{\brp}$ & $\lnot \mathtt{ack\_lost} \Until (\mathtt{success} \land \Eventually \mathtt{failure})$ & 0.37 & \tbcolor{0.084} & 1723 & -2915 & 0.010 & 0.076 \\
  $\formula^{10}_{\brp}$ & $\lnot \mathtt{success} \Until (\mathtt{failure} \land \Eventually \mathtt{uncertain})$ & 0.36 & \tbcolor{0.038} & 1723 & -1467 & 0.0096 & 0.030 \\
  $\formula^{11}_{\brp}$ & $\formula^6_{\brp} \land \formula^7_{\brp} \land \formula^8_{\brp} \land \formula^9_{\brp} \land \formula^{10}_{\brp}$ & 0.48 & \tbcolor{0.057} & 1723 & -1467 & 0.030 & 0.029 \\
  \midrule
  $\formula^{1}_{\crowds}$ & $\Eventually \mathtt{positive}_0$ & 16.07 & \tbcolor{5.69} & 87360 & 780640 & 1.36 & 4.33 \\
  $\formula^{2}_{\crowds}$ & $\bigvee_{i=1}^{9} \Eventually \mathtt{positive}_i$ & 16.09 & \tbcolor{14.18} & 87360 & 780640 & 1.34 & 12.83 \\
  $\formula^{3}_{\crowds}$ & $\Eventually \mathtt{positive}_0 \land \bigwedge_{i=1}^{9} \Globally \lnot \mathtt{positive}_i$ & 15.99 & \tbcolor{3.54} & 30480 & 278512 & 0.90 & 2.63 \\
  $\formula^{4}_{\crowds}$ & $\Eventually \left( \mathtt{new\_run} \land \lnot \Next (\mathtt{latest\_observe}_0 \Until \mathtt{new\_run}) \right)$ & \tbcolor{15.79} & 20.75 & 0 & 0 & 0.52 & 20.23 \\
  $\formula^{5}_{\crowds}$ & $\Globally \left( \lnot \mathtt{new\_run} \lor \Next (\mathtt{latest\_observe}_0 \Until \mathtt{new\_run}) \right)$ & 15.81 & \tbcolor{1.60} & 0 & 0 & 0.54 & 1.06 \\
  \midrule
  $\formula^1_{\egl}$ & $\Eventually (\lnot \mathtt{knowA} \land \mathtt{knowB})$ & 6.70 & \tbcolor{1.95} & 1023 & -1023 & 0.086 & 1.87 \\
  $\formula^2_{\egl}$ & $\Eventually (\mathtt{knowA} \land \lnot \mathtt{knowB})$ & 6.68 & \tbcolor{1.81} & 1023 & -1023 & 0.086 & 1.72 \\
  $\formula^3_{\egl}$ & $\Eventually \mathtt{knowA}$ & 6.70 & \tbcolor{3.33} & 1023 & -1023 & 0.084 & 3.25 \\
  $\formula^4_{\egl}$ & $\Eventually \mathtt{knowB}$ & 6.71 & \tbcolor{3.26} & 1023 & -1023 & 0.084 & 3.18 \\
  $\formula^5_{\egl}$ & $\Eventually \mathtt{knowA} \land \Eventually \mathtt{knowB}$ & 6.74 & \tbcolor{3.19} & 1023 & -1023 & 0.087 & 3.10 \\
  $\formula^6_{\egl}$ & $\Eventually \mathtt{knowA} \land \Globally \lnot \mathtt{knowB}$ & 6.78 & \tbcolor{0.44} & 1023 & -1023 & 0.090 & 0.35 \\
  $\formula^7_{\egl}$ & $\Eventually \mathtt{knowB} \land \Globally \lnot \mathtt{knowA}$ & 6.74 & \tbcolor{0.37} & 1023 & -1023 & 0.090 & 0.28 \\
  $\formula^8_{\egl}$ & $\lnot \mathtt{knowB} \Until \mathtt{knowA}$ & 6.71 & \tbcolor{1.81} & 1023 & -1023 & 0.086 & 1.72 \\
  \midrule
  $\formula^1_{\leader}$ & $\Eventually \mathtt{elected}$ & \tbcolor{0.27} & 0.44 & 17 & 126835 & 0.20 & 0.24 \\
  $\formula^2_{\leader}$ & $\Eventually_{\leq 5} \mathtt{elected}$ & 0.28 & \tbcolor{0.27} & 0 & 0 & 0.0088 & 0.26 \\
  $\formula^3_{\leader}$ & $\Eventually_{\leq 10} \mathtt{elected}$ & \tbcolor{0.28} & 0.45 & 0 & 0 & 0.010 & 0.44 \\
  \midrule
  $\formula^1_{\nand}$ & $\Eventually \mathtt{reliable}$ & 1055.41 & \tbcolor{493.50} & 10946130 & -10957110 & 44.50 & 448.86 \\
  $\formula^2_{\nand}$ & $\Eventually \mathtt{perfect}$ & 1063.07 & \tbcolor{473.29} & 10946130 & -20510955 & 45.41 & 427.69 \\
  $\formula^3_{\nand}$ & $\lnot \mathtt{done} \Until \mathtt{boundary}$ & 1055.90 & \tbcolor{492.91} & 10946130 & -10957110 & 44.58 & 448.19 \\
  \midrule
  $\formula^1_{\haddad}$ & $\Eventually \mathtt{target}$ & \tbcolor{0.032} & 0.041 & 97 & -191 & 0.0053 & 0.037 \\
  $\formula^2_{\haddad}$ & $\Eventually \mathtt{target} \land \Globally \lnot \mathtt{done}$ & 0.041 & \tbcolor{0.029} & 95 & -185 & 0.0052 & 0.024 \\
  $\formula^3_{\haddad}$ & $\Eventually \mathtt{target} \land \Eventually \mathtt{done}$ & \tbcolor{0.041} & 0.041 & 97 & -191 & 0.0053 & 0.038 \\
  $\formula^4_{\haddad}$ & $\lnot \mathtt{bad} \Until (\mathtt{target} \land \Eventually \mathtt{done})$ & 0.041 & \tbcolor{0.040} & 97 & -191 & 0.0052 & 0.036 \\
  $\formula^5_{\haddad}$ & $\lnot \mathtt{done} \Until (\mathtt{target} \land \Eventually \mathtt{center})$ & 0.040 & \tbcolor{0.030} & 97 & -191 & 0.0053 & 0.024 \\
  $\formula^6_{\haddad}$ & $\lnot \mathtt{target} \Until (\mathtt{done} \land \Eventually \mathtt{bad})$ & 0.040 & \tbcolor{0.040} & 97 & -191 & 0.0053 & 0.036 \\
  $\formula^7_{\haddad}$ & $\lnot \mathtt{center} \Until (\mathtt{not\_bad} \land \Eventually \mathtt{bad})$ & 0.041 & \tbcolor{0.040} & 97 & -191 & 0.0053 & 0.037 \\
  $\formula^8_{\haddad}$ & $\lnot \mathtt{center} \Until (\mathtt{left\_or\_center} \land \Eventually \mathtt{target})$ & 0.041 & \tbcolor{0.040} & 97 & -191 & 0.0055 & 0.037 \\
  $\formula^9_{\haddad}$ & $\formula^3_{\haddad} \land \formula^4_{\haddad} \land \formula^5_{\haddad} \land \formula^6_{\haddad}$ & 0.045 & \tbcolor{0.030} & 97 & -191 & 0.0071 & 0.023 \\
  $\formula^{10}_{\haddad}$ & $\formula^3_{\haddad} \land \formula^4_{\haddad} \land \formula^5_{\haddad} \land \formula^6_{\haddad} \land \formula^7_{\haddad}$ & 0.076 & \tbcolor{0.031} & 97 & -191 & 0.012 & 0.022 \\
  \bottomrule
 \end{tabular}}
\end{table*}

\brp{} is a model of the bounded retransmission protocol~\cite{DBLP:conf/types/HelminkSV93}, which is a communication protocol for sending files over a lossy channel with a bounded number of retransmissions.
$\formula^{1}_{\brp}$--$\formula^{3}_{\brp}$ are taken from~\cite{DBLP:conf/papm/DArgenioJJL01}, while $\formula^{4}_{\brp}$--$\formula^{11}_{\brp}$ are our original ones.

\crowds{} is a model of the Crowds protocol~\cite{DBLP:journals/tissec/ReiterR98}, an anonymity protocol for web browsing.
$\formula^1_{\crowds}$--$\formula^3_{\crowds}$ are taken from~\cite{KwiatkowskaNP12}, while $\formula^4_{\crowds}$ and $\formula^{5}_{\crowds}$ are our original ones.

\egl{} is a model of a probabilistic contract signing protocol~\cite{DBLP:journals/cacm/EvenGL85}. %
$\formula^1_{\egl}$ and $\formula^2_{\egl}$ are taken from~\cite{KwiatkowskaNP12}, while $\formula^3_{\egl}$--$\formula^8_{\egl}$ are our original ones.

\leader{} is a model of a synchronous leader election protocol~\cite{DBLP:journals/iandc/ItaiR90}.
All the formulas are taken from~\cite{KwiatkowskaNP12}.

\nand{} is a model of NAND multiplexing, a technique for constructing reliable circuits from unreliable components~\cite{DBLP:journals/tcad/NormanPKS05}.
$\formula^1_{\nand}$ is taken from~\cite{KwiatkowskaNP12}, while $\formula^2_{\nand}$ and $\formula^{3}_{\nand}$ are our original ones.

\haddad{} is an MC used in~\cite{HaddadM18} to motivate their model checking algorithm.
$\formula^1_{\haddad}$ is taken from~\cite{HaddadM18}, while $\formula^2_{\haddad}$--$\formula^{10}_{\haddad}$ are our original ones.

\subsection{Experiments}\label{section:experimens:setting}

We used \storm{}~\cite{HenselJKQV22} version \revision{1.13.0} as the baseline model checker.
When measuring the execution time of \storm{}, we directly executed \storm{} for the original MC, while in measuring the execution time of our approach,
we executed \ourTool{}, which invokes \storm{} as a child process.
\revision{We invoked \storm{}'s exact verifier for \haddad{}, and we invoked \storm{} with the default configuration for the other benchmarks.}

\revision{\ourTool{} currently} accepts only MCs in the DRN format, which is an explicit format specific to \storm{}.
We converted the original MCs in the PRISM format to the DRN format using \storm{} beforehand, and the resulting DRN file was used as the input for both \storm{} and \ourTool{}.
We constructed a DRN file for each MC such that the DRN file contains the labels used in the LTL formulas.

\revision{All experiments were conducted on a computational server with Intel Xeon Platinum 8592V and 1007 GiB of RAM running Ubuntu 24.04.4 LTS.\@
In the experiments, we limited CPU usage to at most four cores and memory usage to 16 GiB.}\@
When measuring the \revision{preprocessing time}, we included the time required to construct the DRA and the product MC.\@

\subsection{RQ1: \revision{Comparison with \storm{}}}\label{section:experimens:model_checking_performance}

To answer RQ1, we \revision{compared the performance of \ourTool{} and \storm{} on the benchmarks described in \cref{section:experimens:benchmarks}.}
\revision{For both \ourTool{} and \storm{},} we executed each benchmark 30 times and measured the mean execution time.
In the bounded MCCS construction, we used the bound $\MCCSBound = 3$.
\cref{table:properties_brp,table:properties_crowds,table:properties_egl,table:properties_haddad,table:properties_leader,table:properties_nand} summarizes the results.

In \cref{table:properties_brp,table:properties_crowds,table:properties_egl,table:properties_haddad,table:properties_leader,table:properties_nand}, we observe that \revision{\ourTool{} is faster than \storm{} in most cases}.
The improvement is particularly evident \revision{when our abstraction significantly reduces the state space of the MC (\eg{} $\nand$) or the verified property has} multiple temporal operators \revision{(\eg{} $\formula^{11}_{\brp}$, $\formula^{5}_{\crowds}$\LongVersion{, $\formula^{7}_{\egl}$}, and $\formula^{10}_{\haddad}$).}
\revision{Since the cost of probabilistic model checking is sensitive to the size of the verified MC,}
\revision{reducing} the state space can \revision{significantly} decrease the execution time of model checking. %

Another reason \revision{for this improvement} is that even if the abstraction does not reduce the state space, %
\revision{\ourTool{}} reduces the number of labels in the MC given to \storm{}, which often improves the efficiency.
Namely, the resulting state-labeled product MC only has labels for encoding the Rabin acceptance condition; the number of labels is at most two in our benchmarks.

However, a speedup in the model-checking phase does not always reduce the end-to-end runtime due to the abstraction's overhead.
For instance, in \revision{$\formula^1_{\leader}$}, the execution time of model checking decreased from \revision{0.27} seconds to \revision{0.24} seconds, but the total execution time increased to \revision{0.44} seconds due to the overhead of specification-guided \revision{path shortcutting}.
Still, the execution time of the preprocessing is less than \revision{1.4} seconds in all the benchmarks except for \nand{}, where the MC is huge.
We believe that this overhead is acceptable in practice.

We also observe that there are some exceptional cases where \revision{the preprocessing in \ourTool{}} increased the execution time of model checking, such as $\formula^{4}_{\crowds}$ and \revision{$\formula^{3}_{\leader}$}.
This is partly because of the complexity of the LTL formula encoding the Rabin acceptance condition.
For instance, in $\formula^{4}_{\crowds}$, the Rabin acceptance condition is encoded by an LTL formula $(\Globally \Eventually L)  \land (\Eventually \Globally \lnot U)$, which is likely more challenging than the original formula for \storm{} to verify.
Since the number of states is not reduced in $\formula^{4}_{\crowds}$, the model checking algorithm took more time.
In contrast, although the number of states is also not reduced in $\formula^{5}_{\crowds}$, the Rabin acceptance condition is encoded by a simpler LTL formula $\Globally \Eventually L$, and thus, the model checking algorithm took less time.

Another reason for the increase in execution time is the state-space blow-up in the product MC.
For $\formula^3_{\leader}$, the time bound in the LTL formulas (\ie{} $\leq 10$) is encoded by the state space of the DRAs, and thus, the product MC has a larger state space than the original MC, which likely caused the increase in execution time.

Overall, we answer RQ1 as follows:

\rqanswer{RQ1}{%
\revision{\ourTool{} usually reduces end-to-end execution time, particularly when specification-guided path shortcutting substantially reduces the state space.}
}

\subsection{\revision{RQ2: Isolating the Contribution of Path Shortcutting}}\label{section:experimens:ablation_study}

\reviewer{(meta)}{Ablation: Isolate the contribution of specification-guided abstraction from confounding factors (DRA product encoding, label reduction). The current experiments do not cleanly attribute performance gains.}

\begin{table*}[tbp]
 \centering
 \caption{Results of the experiments conducted to answer RQ2 and RQ3. The columns have the same meaning as in \cref{table:properties_brp}. For ``No Abstraction'', the total time shorter than ``Ours'' in \cref{table:properties_brp} is highlighted. For ``$\ast$+Bism'', the total time shorter than ``$\ast$'' in \cref{table:properties_brp} or this table is highlighted, where ${\ast} \in \{\storm{}, \text{``Ours''}, \text{``No Abstraction''}\}$.}%
 \label{table:ablation_study_results}
 \LongVersion{\small}\ShortVersion{\scriptsize}
 \revision{\begin{tabular}{lrrr|r|rrr|rrr}
  \toprule
  & \multicolumn{3}{c|}{No Abstraction} & \multicolumn{1}{c|}{\storm{}+Bisim} & \multicolumn{3}{c|}{Ours+Bisim} & \multicolumn{3}{c}{No Abstraction+Bisim} \\
  & Total & Prep. & Storm & Total & Total & Prep. & Storm & Total & Prep. & Storm \\
  & Time  & Time  & Time  & Time  & Time  & Time  & Time  & Time  & Time  & Time \\
  \midrule
$\formula^{1}_{\brp}$ & \tbcolor 0.076 & 0.0080 & 0.069 & \tbcolor 0.37 & 0.079 & 0.0096 & 0.071 & 0.077 & 0.0081 & 0.070 \\
$\formula^{2}_{\brp}$ & \tbcolor 0.051 & 0.0078 & 0.046 & 0.36 & \tbcolor 0.053 & 0.0095 & 0.045 & 0.053 & 0.0079 & 0.047 \\
$\formula^{3}_{\brp}$ & 0.35 & 0.0076 & 0.34 & 0.34 & \tbcolor 0.25 & 0.0091 & 0.24 & \tbcolor 0.34 & 0.0078 & 0.34 \\
$\formula^{4}_{\brp}$ & 0.12 & 0.0079 & 0.12 & 0.36 & 0.096 & 0.0091 & 0.090 & 0.12 & 0.0078 & 0.12 \\
$\formula^{5}_{\brp}$ & 0.35 & 0.0086 & 0.34 & 0.35 & 0.26 & 0.0099 & 0.25 & \tbcolor 0.35 & 0.0087 & 0.34 \\
$\formula^{6}_{\brp}$ & \tbcolor 0.076 & 0.0094 & 0.068 & 0.40 & 0.086 & 0.011 & 0.076 & 0.087 & 0.0093 & 0.080 \\
$\formula^{7}_{\brp}$ & \tbcolor 0.048 & 0.0084 & 0.040 & 0.38 & \tbcolor 0.048 & 0.010 & 0.039 & 0.048 & 0.0083 & 0.042 \\
$\formula^{8}_{\brp}$ & 0.35 & 0.0091 & 0.34 & 0.36 & 0.27 & 0.010 & 0.26 & 0.35 & 0.0091 & 0.35 \\
$\formula^{9}_{\brp}$ & \tbcolor 0.082 & 0.0087 & 0.075 & \tbcolor 0.37 & 0.092 & 0.010 & 0.084 & 0.094 & 0.0087 & 0.088 \\
$\formula^{10}_{\brp}$ & \tbcolor 0.035 & 0.0080 & 0.030 & 0.37 & \tbcolor 0.036 & 0.0095 & 0.030 & 0.036 & 0.0079 & 0.031 \\
$\formula^{11}_{\brp}$ & 0.057 & 0.028 & 0.030 & 0.48 & 0.059 & 0.030 & 0.031 & \tbcolor 0.057 & 0.028 & 0.031 \\\hline
$\formula^{1}_{\crowds}$ & \tbcolor 3.52 & 0.40 & 3.12 & \tbcolor 15.61 & \tbcolor 5.56 & 1.37 & 4.19 & \tbcolor 3.23 & 0.40 & 2.83 \\
$\formula^{2}_{\crowds}$ & 14.61 & 0.40 & 14.21 & 16.14 & \tbcolor 14.17 & 1.36 & 12.80 & \tbcolor 14.44 & 0.40 & 14.03 \\
$\formula^{3}_{\crowds}$ & \tbcolor 2.75 & 0.41 & 2.33 & 16.45 & \tbcolor 3.46 & 0.90 & 2.55 & \tbcolor 2.60 & 0.41 & 2.19 \\
$\formula^{4}_{\crowds}$ & \tbcolor 20.56 & 0.40 & 20.16 & \tbcolor 15.77 & \tbcolor 20.67 & 0.52 & 20.14 & 20.79 & 0.41 & 20.37 \\
$\formula^{5}_{\crowds}$ & \tbcolor 1.48 & 0.41 & 1.06 & \tbcolor 15.74 & \tbcolor 1.55 & 0.54 & 1.01 & \tbcolor 1.43 & 0.41 & 1.02 \\\hline
$\formula^{1}_{\egl}$ & \tbcolor 1.93 & 0.069 & 1.86 & 6.71 & \tbcolor 1.93 & 0.085 & 1.85 & \tbcolor 1.92 & 0.070 & 1.85 \\
$\formula^{2}_{\egl}$ & \tbcolor 1.79 & 0.070 & 1.72 & 6.71 & \tbcolor 1.79 & 0.084 & 1.71 & \tbcolor 1.78 & 0.070 & 1.70 \\
$\formula^{3}_{\egl}$ & \tbcolor 3.31 & 0.067 & 3.25 & \tbcolor 6.69 & 3.47 & 0.082 & 3.39 & 3.46 & 0.067 & 3.39 \\
$\formula^{4}_{\egl}$ & \tbcolor 3.24 & 0.068 & 3.17 & \tbcolor 6.69 & \tbcolor 3.25 & 0.082 & 3.17 & \tbcolor 3.22 & 0.067 & 3.16 \\
$\formula^{5}_{\egl}$ & \tbcolor 3.10 & 0.071 & 3.03 & 6.79 & \tbcolor 3.11 & 0.087 & 3.02 & \tbcolor 3.09 & 0.071 & 3.02 \\
$\formula^{6}_{\egl}$ & \tbcolor 0.41 & 0.070 & 0.35 & \tbcolor 6.76 & \tbcolor 0.43 & 0.086 & 0.35 & 0.42 & 0.071 & 0.35 \\
$\formula^{7}_{\egl}$ & \tbcolor 0.34 & 0.071 & 0.27 & 6.92 & \tbcolor 0.35 & 0.086 & 0.27 & \tbcolor 0.34 & 0.072 & 0.27 \\
$\formula^{8}_{\egl}$ & \tbcolor 1.80 & 0.075 & 1.72 & 6.73 & \tbcolor 1.79 & 0.084 & 1.71 & \tbcolor 1.78 & 0.075 & 1.70 \\\hline
$\formula^{1}_{\leader}$ & \tbcolor 0.042 & 0.0066 & 0.038 & \tbcolor 0.26 & \tbcolor 0.43 & 0.20 & 0.23 & 0.042 & 0.0068 & 0.038 \\
$\formula^{2}_{\leader}$ & \tbcolor 0.27 & 0.0083 & 0.26 & \tbcolor 0.27 & \tbcolor 0.27 & 0.0093 & 0.26 & 0.27 & 0.0082 & 0.26 \\
$\formula^{3}_{\leader}$ & \tbcolor 0.44 & 0.0095 & 0.44 & \tbcolor 0.27 & \tbcolor 0.44 & 0.010 & 0.43 & 0.44 & 0.0094 & 0.43 \\\hline
$\formula^{1}_{\nand}$ & 1064.97 & 20.35 & 1044.42 & 1055.98 & \tbcolor 493.48 & 44.63 & 448.68 & \tbcolor 1058.30 & 20.33 & 1037.75 \\
$\formula^{2}_{\nand}$ & 988.00 & 22.57 & 965.22 & 1183.87 & 645.97 & 45.08 & 600.68 & 1106.76 & 22.48 & 1084.05 \\
$\formula^{3}_{\nand}$ & 1062.92 & 20.44 & 1042.28 & 1058.56 & 494.15 & 45.82 & 448.15 & \tbcolor 1059.54 & 20.57 & 1038.75 \\\hline
$\formula^{1}_{\haddad}$ & 0.041 & 0.0048 & 0.038 & 0.033 & \tbcolor 0.040 & 0.0054 & 0.036 & \tbcolor 0.040 & 0.0049 & 0.038 \\
$\formula^{2}_{\haddad}$ & \tbcolor 0.028 & 0.0046 & 0.024 & 0.041 & \tbcolor 0.029 & 0.0051 & 0.023 & 0.028 & 0.0046 & 0.024 \\
$\formula^{3}_{\haddad}$ & 0.042 & 0.0046 & 0.039 & 0.041 & \tbcolor 0.041 & 0.0051 & 0.037 & 0.042 & 0.0047 & 0.039 \\
$\formula^{4}_{\haddad}$ & 0.041 & 0.0048 & 0.038 & \tbcolor 0.040 & \tbcolor 0.040 & 0.0052 & 0.036 & \tbcolor 0.040 & 0.0048 & 0.038 \\
$\formula^{5}_{\haddad}$ & \tbcolor 0.029 & 0.0048 & 0.024 & 0.040 & 0.030 & 0.0053 & 0.024 & 0.030 & 0.0048 & 0.024 \\
$\formula^{6}_{\haddad}$ & 0.040 & 0.0047 & 0.038 & 0.040 & 0.040 & 0.0051 & 0.036 & 0.040 & 0.0047 & 0.038 \\
$\formula^{7}_{\haddad}$ & 0.041 & 0.0048 & 0.038 & \tbcolor 0.040 & 0.041 & 0.0053 & 0.037 & 0.041 & 0.0048 & 0.039 \\
$\formula^{8}_{\haddad}$ & 0.040 & 0.0049 & 0.038 & 0.041 & 0.040 & 0.0055 & 0.037 & 0.041 & 0.0050 & 0.038 \\
$\formula^{9}_{\haddad}$ & \tbcolor 0.028 & 0.0062 & 0.023 & \tbcolor 0.044 & \tbcolor 0.028 & 0.0068 & 0.023 & 0.028 & 0.0062 & 0.023 \\
$\formula^{10}_{\haddad}$ & 0.031 & 0.011 & 0.022 & 0.076 & 0.032 & 0.012 & 0.022 & \tbcolor 0.031 & 0.011 & 0.022 \\
  \bottomrule
 \end{tabular}}
\end{table*}

\revision{%
To answer RQ2, we measured the execution time of \ourTool{} with specification-guided path shortcutting \emph{disabled}.
\LongVersion{Namely, we measured the execution time of a variant of \ourTool{} that
\begin{ienumeration}
 \item translates $\tilde\MC$ into a transition-labeled MC $\MC$;
 \item constructs a DRA $\automaton_{\formula}$ from $\formula$;
 \item eliminates the labels on the edges of the product MC $\MC \otimes \automaton_{\formula}$ to obtain an unlabeled MC; and
 \item invokes \storm{}~\cite{HenselJKQV22} on the product MC together with an LTL formula $\formula'$ encoding the acceptance condition of $\automaton_{\formula}$.
\end{ienumeration}}
We executed this workflow 30 times for each benchmark and measured the mean execution time.
The columns ``No Abstraction'' in \cref{table:ablation_study_results} summarize the results.

In \cref{table:properties_brp,table:ablation_study_results}, we observe that ``No Abstraction'' is typically faster than \storm{}, which suggests that the DRA product construction and the label reduction contribute to improving the efficiency of model checking.
Moreover, ``No Abstraction'' is often even faster than ``Ours''.
This is partly because of the overhead of specification-guided path shortcutting, as observed in, \eg{} $\formula^7_{\brp}$.

Interestingly, specification-guided \revision{path shortcutting} can increase the execution time of subsequent model checking even when it reduces the state space.
This is particularly evident in $\formula^1_{\leader}$, where shortcutting removed 17 states, but the execution time of model checking increased from 0.038 seconds to 0.24 seconds.
This is likely because the abstraction increased the number of edges in the product MC from 5,268 to 132,103.
The number of edges can increase when we remove a state with many incoming and outgoing edges.
Nevertheless, the overhead of specification-guided \revision{path shortcutting} is typically much smaller than the execution time of model checking, and the overall execution times of ``Ours'' and ``No Abstraction'' are usually similar even if the abstraction increases the execution time of model checking.

In contrast, \ourTool{} tends to be significantly faster than ``No Abstraction'' when the state space is significantly reduced.
For instance, in $\formula^{1}_{\nand}$, the number of states is reduced from 18,826,082 to 7,879,952, and the end-to-end execution time is reduced from 1064.97 seconds to 493.50 seconds.
For the other benchmarks, the state space reduction is not significant, but the end-to-end execution time can still be reduced.
For instance, in $\formula^{3}_{\brp}$, the number of states is reduced from 5,192 to 3,469, and the end-to-end execution time is reduced from 0.35 seconds to 0.25 seconds.

Overall, we answer RQ2 as follows:

\rqanswer{RQ2}{%
\revision{Although DRA product construction and label reduction can already improve the efficiency of model checking, specification-guided path shortcutting can further improve the efficiency in some cases, especially when it significantly reduces the state space.}}
}

\subsection{\revision{RQ3: Combination with Bisimulation Minimization}}\label{section:experimens:bisimulation_minimization}

\reviewer{meta}{Preprocessing: Add empirical comparison against preprocessing baselines to assess the relative contribution of specification-guided path shortcutting. The authors committed to this in their response.}

\revision{%
To answer RQ3, we measured the execution time of \storm{}, \ourTool{}, and ``No Abstraction'' with \storm{}'s bisimulation minimization~\cite{KatoenKZJ07} enabled.
We executed this workflow 30 times for each benchmark and measured the mean execution time.
The columns with ``+Bisim'' in \cref{table:ablation_study_results} summarize the results.

In \cref{table:properties_brp,table:ablation_study_results}, we observe that bisimulation minimization seems to improve the efficiency of \ourTool{} more often (25 properties) than \storm{} (14 properties) and ``No Abstraction'' (18 properties).
This can be partly because specification-guided path shortcutting abstracts away differences between states that are irrelevant to the verified property, and more states can be merged by bisimulation minimization.

We also observe that the effect of bisimulation minimization on the efficiency of model checking is usually not as significant as that of specification-guided \revision{path shortcutting}.
For instance, in $\formula^{1}_{\nand}$, the end-to-end execution time of ``No Abstraction'' is reduced from 1064.97 seconds to 1058.30 seconds by bisimulation minimization, while it is reduced to 493.48 seconds by combining specification-guided path shortcutting.
This is likely because many states of an MC are not bisimilar to each other, and thus, bisimulation minimization cannot significantly reduce the state space in many cases.

Overall, we answer RQ3 as follows:

\rqanswer{RQ3}{%
\revision{The use of bisimulation minimization can further improve the efficiency of \ourTool{}, while its influence is usually not as significant as that of specification-guided path shortcutting.}}
}

\subsection{\revision{RQ4}: Sensitivity to the MCCS Bound $\MCCSBound$}\label{section:experimens:parameter_sensitivity}

\reviewer{meta}{Choice of $\MCCSBound$: Provide practical guidance for selecting the MCCS bound. The current justification is insufficient. Extend the sensitivity evaluation to a broader range of $\MCCSBound$ values.}

\begin{table}[tb]
 \centering
 \caption{Mean execution time [sec.] and the number of states eliminated by specification-guided \revision{path shortcutting} for $\formula^{11}_{\brp}$, $\formula^{2}_{\crowds}$, $\formula^{5}_{\crowds}$, $\formula^{3}_{\egl}$, $\formula^{1}_{\leader}$, $\formula^{1}_{\nand}$, and $\formula^{2}_{\haddad}$ with different values of the MCCS bound $\MCCSBound$. Each entry reports the mean over 30 runs. The best time for each formula (including $\MCCSBound = 3$) is highlighted.}%
 \label{table:results_parameter_sensitivity}
 \ShortVersion{\scriptsize}\LongVersion{\small}
 \revision{
\begin{tabular}{llrrrr}
\toprule
  & \multirow{2}{*}{$\MCCSBound$} & Total & Reduced & Abstraction & Storm \\
  & & Time & States & Time & Time \\
\midrule
\multirow{6}{*}{$\formula^{11}_{\brp}$} & 2 & 0.084 & 1723 & 0.052 & 0.034 \\
 & 3 & 0.092 & 1723 & 0.051 & 0.042 \\
 & 4 & 0.073 & 1723 & 0.041 & 0.034 \\
 & 5 & \tbcolor{0.064} & 1723 & 0.033 & 0.033 \\
 & 8 & 0.065 & 1723 & 0.035 & 0.032 \\
 & 13 & 0.065 & 1723 & 0.037 & 0.030 \\
\midrule
\multirow{6}{*}{$\formula^{2}_{\crowds}$} & 2 & 14.54 & 87360 & 1.67 & 12.86 \\
 & 3 & \tbcolor{14.48} & 87360 & 1.62 & 12.85 \\
 & 4 & 14.56 & 87360 & 1.69 & 12.86 \\
 & 5 & 14.58 & 87360 & 1.70 & 12.88 \\
 & 8 & 14.61 & 87360 & 1.39 & 13.22 \\
 & 13 & 14.58 & 87360 & 1.55 & 13.02 \\
\midrule
\multirow{6}{*}{$\formula^{5}_{\crowds}$} & 2 & 1.70 & 0 & 0.62 & 1.08 \\
 & 3 & \tbcolor{1.67} & 0 & 0.58 & 1.09 \\
 & 4 & 1.82 & 0 & 0.72 & 1.10 \\
 & 5 & 1.81 & 0 & 0.72 & 1.08 \\
 & 8 & 2.15 & 0 & 1.07 & 1.07 \\
 & 13 & 29.25 & 0 & 28.19 & 1.06 \\
\midrule
\multirow{6}{*}{$\formula^{3}_{\egl}$} & 2 & \tbcolor{3.34} & 1023 & 0.086 & 3.25 \\
 & 3 & 3.36 & 1023 & 0.096 & 3.26 \\
 & 4 & 3.36 & 1023 & 0.087 & 3.27 \\
 & 5 & \tbcolor{3.34} & 1023 & 0.086 & 3.26 \\
 & 8 & 3.35 & 1023 & 0.088 & 3.26 \\
 & 13 & 3.35 & 1023 & 0.091 & 3.26 \\
\midrule
\multirow{6}{*}{$\formula^{1}_{\leader}$} & 2 & 0.47 & 17 & 0.22 & 0.25 \\
 & 3 & \tbcolor{0.46} & 17 & 0.21 & 0.25 \\
 & 4 & 0.46 & 17 & 0.21 & 0.25 \\
 & 5 & 0.48 & 17 & 0.22 & 0.26 \\
 & 8 & 0.47 & 17 & 0.22 & 0.26 \\
 & 13 & 0.47 & 17 & 0.21 & 0.25 \\
\midrule
\multirow{6}{*}{$\formula^{1}_{\nand}$} & 2 & \tbcolor{493.68} & 10946130 & 45.05 & 448.45 \\
 & 3 & 493.79 & 10946130 & 45.25 & 448.37 \\
 & 4 & 494.68 & 10946130 & 45.82 & 448.66 \\
 & 5 & 494.64 & 10946130 & 45.81 & 448.63 \\
 & 8 & 495.05 & 10946130 & 45.62 & 449.24 \\
 & 13 & 493.97 & 10946130 & 45.30 & 448.48 \\
\midrule
\multirow{6}{*}{$\formula^{2}_{\haddad}$} & 2 & 0.041 & 95 & 0.0088 & 0.033 \\
 & 3 & \tbcolor{0.029} & 95 & 0.0067 & 0.025 \\
 & 4 & 0.034 & 95 & 0.0071 & 0.027 \\
 & 5 & 0.038 & 95 & 0.0084 & 0.031 \\
 & 8 & 0.041 & 95 & 0.0086 & 0.033 \\
 & 13 & 0.037 & 95 & 0.0079 & 0.028 \\
\bottomrule
\end{tabular}}
\end{table}

To answer \revision{RQ4}, we \revision{evaluated $\MCCSBound \in \{2,3,4,5,8,13\}$ on seven} formulas with varying levels of complexity and measured the mean execution time over 30 runs.
\cref{table:results_parameter_sensitivity} summarizes the results.
For $\MCCSBound = 3$, we reran the experiments instead of reusing the measurements from RQ1 so that \cref{table:results_parameter_sensitivity} is based on an independent set of runs.

In \cref{table:results_parameter_sensitivity}, changing $\MCCSBound$ from 2 to \revision{13} does not affect the number of reduced states for these \revision{seven} formulas.
\revision{For most of the formulas, the choice of $\MCCSBound$ has little effect on runtime.
In contrast, for $\formula^{5}_{\crowds}$, the runtime increased substantially when $\MCCSBound$ was increased, which is likely because the algorithm must perform deeper exploration when trying to find MCCSs.}
This suggests that, at least for these benchmarks, the MCCSs needed for effective abstraction are typically short, and a moderate bound is sufficient.
\revision{Therefore, trying a small to moderate bound (\eg{} $\MCCSBound = 2$ or $\MCCSBound = 3$) is likely a good choice in practice.}

Overall, we answer \revision{RQ4} as follows:

\rqanswer{\revision{RQ4}}{%
On the selected benchmarks, varying the MCCS bound $\MCCSBound$ did not change the number of reduced states, suggesting that a small to moderate bound \revision{(\eg{} $\MCCSBound = 2$ or $\MCCSBound = 3$)} is sufficient in practice.}

\subsection{Potential Usage Beyond Efficient Model Checking}\label{section:experimens:interpretability}

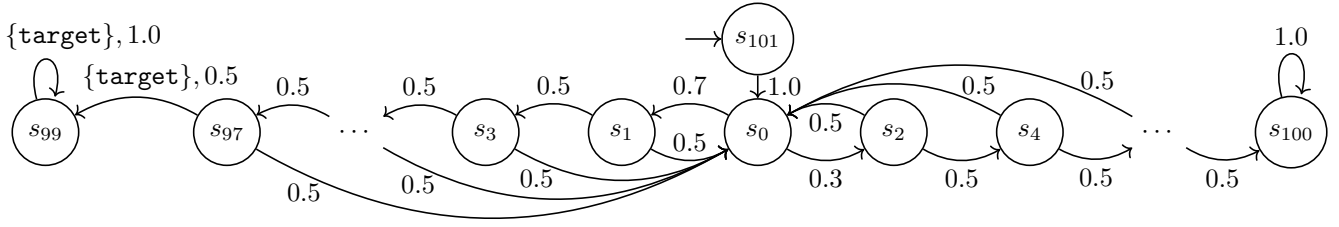
\begin{figure*}[tb]
 \centering
 \begin{tikzpicture}[auto, node distance=0.9cm,semithick,\ShortVersion{scale=0.7,}every node/.style={initial text={},transform shape}]
  \node[state, initial] (init) {$\state_{101}$};
  \node[state] (qN) [node distance=0.3cm, below=of init]{$\state_{0}$};
  \node[state] (qNp1) [right=of qN]{$\state_{2}$};
  \node[state] (qNp2) [right=of qNp1]{$\state_{4}$};
  \node (qNp3) [right=of qNp2]{$\cdots$};
  \node[state] (q2N) [right=of qNp3]{$\state_{100}$};
  \node[state] (qNm1) [left=of qN]{$\state_{1}$};
  \node[state] (qNm2) [left=of qNm1]{$\state_{3}$};
  \node (qNm3) [left=of qNm2]{$\cdots$};
  \node[state] (q1) [left=of qNm3]{$\state_{97}$};
  \node[state] (q0) [node distance=1.5cm,left=of q1]{$\state_{99}$};

  \path[->]
  (init) edge node[right] {$1.0$} (qN)
  (qN) edge[bend right] node[below] {$0.3$} (qNp1)
  (qN) edge[bend right] node[above] {$0.7$} (qNm1)
  (qNp1) edge[bend right] node[below] {$0.5$} (qNp2)
  (qNp1) edge[bend right] node[below] {$0.5$} (qN)
  (qNp2) edge[bend right] node[above,pos=0.1] {$0.5$} (qN)
  (qNp2) edge[bend right] node[below] {$0.5$} (qNp3)
  (qNp3) edge[bend right] node[above,pos=0.1] {$0.5$} (qN)
  (qNm1) edge[bend right] node[above] {$0.5$} (qNm2)
  (qNm1) edge[bend right] node[above] {$0.5$} (qN)
  (qNm2) edge[bend right] node[below,pos=0.1] {$0.5$} (qN)
  (qNm2) edge[bend right] node[above] {$0.5$} (qNm3)
  (qNm3) edge[bend right] node[below,pos=0.1] {$0.5$} (qN)
  (qNm3) edge[bend right] node[above] {$0.5$} (q1)
  (q1) edge[bend right] node[below,pos=0.1] {$0.5$} (qN)
  (qNp3) edge[bend right] node[below] {$0.5$} (q2N)
  (q1) edge[bend right] node[above,pos=0.3] {$\{\mathtt{target}\},0.5$} (q0)
  (q0) edge[loop above] node[xshift=0.5cm] {$\{\mathtt{target}\},1.0$} (q0)
  (q2N) edge[loop above] node {$1.0$} (q2N)
  ;
 \end{tikzpicture}
 \caption{The transition-labeled MC in $\haddad$. The initial state $\state_{101}$ was added during the construction from a state-labeled MC, and the number of states is one greater than the number shown in \cref{table:benchmark_summary}. Labels other than $\{\mathtt{target}\}$ are omitted.}%
 \label{figure:mc_haddad}
\end{figure*}

\begin{figure}[tbp]
 \centering
 \begin{tikzpicture}[auto, semithick,node distance=1.5cm,scale=0.80,every node/.style={initial text={},transform shape}]
  \node[state, initial] (q0) at (-3.75,0) {$\state_{101}$};
  \node[state] (q1) at (0,1.5) {$\state_{1}$};
  \node[state] (q2) at (0,-1.5) {$\state_{2}$};
  \node[state] (q3) [node distance=1.7cm,right=of q1] {$\state_{99}$};
  \node[state] (q4) [node distance=1.7cm,right=of q2] {$\state_{100}$};

  \path[->]
  (q0) edge node[above left] {$0.7$} (q1)
  (q0) edge node[below left] {$0.3$} (q2)
  (q1) edge[loop left] node {$0.7 \times \bigl(1 - (1/2)^{49}\bigr)$} (q1)
  (q1) edge[bend left=10] node {$0.3 \times \bigl(1 - (1/2)^{49}\bigr)$} (q2)
  (q1) edge node[above=0.5cm] {$\{\mathtt{target}\}, (1/2)^{49}$} (q3)
  (q2) edge[bend left=10] node {$0.7 \times \bigl(1 - (1/2)^{49}\bigr)$} (q1)
  (q2) edge[loop left] node {$0.3 \times \bigl(1 - (1/2)^{49}\bigr)$} (q2)
  (q2) edge node {$(1/2)^{49}$} (q4)
  (q3) edge[loop right] node[above=0.4cm] {$\{\mathtt{target}\}, 1$} (q3)
  (q4) edge[loop right] node {$1$} (q4)
  ;
 \end{tikzpicture}
 \caption{The MC abstracted with respect to $\formula^{1}_{\haddad}$ with $\MCCSBound = 3$. Labels other than $\{\mathtt{target}\}$ are omitted.}%
 \label{figure:product_mc_haddad}
\end{figure}

\revision{Beyond improving model-checking efficiency}, we observe that specification-guided \revision{path shortcutting} can provide a concise MC focusing on the behaviors relevant to the given specification.
For instance, \cref{figure:mc_haddad} shows the transition-labeled MC for $\haddad$ and \cref{figure:product_mc_haddad} shows the MC abstracted with respect to $\formula^{1}_{\haddad}$ with $\MCCSBound = 3$.
The abstracted MC has only five states and is much easier for humans to interpret than the original MC with 102 states.

The abstracted MC preserves the satisfaction probability of $\formula^{1}_{\haddad}$ and can be used for understanding the behaviors relevant to $\formula^{1}_{\haddad}$.
For instance, we can easily see that once $\mathtt{target}$ holds, the system will stay in the state where $\mathtt{target}$ holds with probability 1.
This abstracted MC also preserves \revision{other aspects of the behavior}.
For instance, from the non-terminal left-hand side states (\ie{} $\state_{1}, \state_3, \dots, \state_{97}$) in \cref{figure:mc_haddad}, which correspond to the state $\state_{1}$ in \cref{figure:product_mc_haddad}, we have a high probability to stay in the same group of non-terminal states, a small probability to move to the other group of non-terminal states, and a very small probability to move to the terminal state.
Although abstracted MCs are not always as small as this illustration,
we believe such abstracted MCs are potentially useful, for example, for understanding and debugging the system.

\section{Related work}\label{section:related}
The work most closely related to ours is that of Matsumoto et al.~\cite{Matsumoto0SW25}, which proposes a black-box checking method based on a specification-guided abstraction for (deterministic) Mealy machines.
Their specification-guided abstraction with respect to an LTL formula constructs quotient Mealy machines tailored to equivalence relations over output \emph{characters} by a generalized $\mathrm{L}^\ast$-algorithm~\cite{Angluin87}.
In contrast, we develop a novel specification-guided \revision{path shortcutting method} for probabilistic systems, namely MCs, based on equivalence relations over \emph{acceptance-preserving words}.

Bisimulation minimization is a well-known preprocessing technique for speeding up probabilistic model checking of Markov chains~\cite{KatoenKZJ07},
and \storm{} supports \emph{symbolic bisimulation minimization}~\cite{Wimmer2010,DijkP18} as an efficient preprocessing~\cite{Hensel18,HenselJKQV22}.
Such \emph{symbolic methods} with decision diagrams for our specification-guided \revision{path shortcutting} would be highly beneficial to speed up probabilistic model checkers; we leave this as future work.

\reviewer{meta}{Related work: Explicitly distinguish the proposed approach from prior specification-guided abstraction work (Dutreix et al; Cleaveland et al.).}

\revision{Cleaveland et al.~\cite{CleavelandRSL22} propose specification-guided abstraction for statistical model checking. Specifically, they introduce a novel sound technique for eliminating some nondeterministic branching from probabilistic automata that abstract target systems, which is beneficial for achieving high performance. Their technique is orthogonal to ours: we aim to reduce the size of the state space by inspecting MCCSs, whereas they aim to remove redundant nondeterminism under their monotonicity assumption. }

\revision{Another well-studied technique for abstracting models with respect to a given specification is \emph{abstraction-refinement} (or \emph{CEGAR}) for probabilistic model checking, as in~\cite{ChadhaV10,DutreixC21}. Our abstraction is a one-shot preprocessing step, whereas these refinement approaches iteratively refine models based on counterexamples found during the model checking phase.  }  

\emph{Compositional probabilistic model checking}~\cite{KwiatkowskaNPQ13,WatanabeEAH23,WatanabeVHRJ24,DelahayeCL10} aims at providing an efficient compositional algorithm by exploiting given compositional structures of stochastic systems such as parallel compositions and sequential compositions.
This exploitation of compositionality is an orthogonal approach to avoid the state-space explosion problem.
One interesting future direction is to lift our specification-guided \revision{path shortcutting} techniques to a compositional specification-guided abstraction that is sound with respect to compositions.

\finalversion{
\emph{Path abstraction}~\cite{AbrahamJWKB10,HartmannsM25} for MCs is an abstraction technique that collapses certain paths into one-step transitions and assigns each resulting transition a probability equal to the sum of the probabilities of the corresponding paths.
Although path abstraction can also reduce the number of states, our novel abstraction is applied before product construction and is guided by a given LTL specification.
}

\section{Conclusion and perspectives}\label{section:conclusion}

We presented a specification-guided \revision{path shortcutting} method for probabilistic model checking of MCs.
Given an LTL formula $\formula$, our abstraction exactly preserves the satisfaction probability of $\formula$, thus serving as a new preprocessing technique prior to running probabilistic model checking algorithms.
We implemented a prototype using \storm{} and Spot, and evaluated its empirical performance against \storm{}.
We demonstrate that our prototype often outperforms the baseline, with especially strong gains on several LTL formulas with multiple temporal operators, as is often the case in the verification of embedded systems that require multiple constraints expressed as conjunctions of specifications.

In addition to the future work mentioned in~\cref{section:related}, we plan to generalize our \revision{abstraction} to other types of $\omega$-automata commonly used in probabilistic model checking (\eg{} limit-deterministic Büchi automata~\cite{SickertEJK16,HahnLST015} and unambiguous Büchi automata~\cite{BaierK00023}).
For instance, the product construction of Markov chains and unambiguous Büchi automata~\cite{BaierK00023} yields weighted systems that may not be stochastic; thus, our current proof, which exploits the properties of BSCCs, is not directly applicable.

It would also be exciting to provide a systematic and uniform framework for generalized specification-guided \revision{path shortcutting} across a variety of systems, including nondeterministic and probabilistic systems \revision{(\eg{} Markov decision processes)}.
Developing a coalgebraic framework for our\LongVersion{ specification-guided \revision{path shortcutting}}\ShortVersion{ \revision{abstraction}}, building on existing coalgebraic product constructions such as~\cite{CirsteaK23,WatanabeJRH25}, would be a promising direction.

\section*{Acknowledgment}

The authors disclose that OpenAI ChatGPT was used for language polishing of selected parts of the manuscript.
OpenAI Codex was used for generating an early version of parts of the implementation used in the experiments.
The generated code was reviewed and verified by the authors.
All scientific content, technical claims, experimental results, and conclusions were verified and finalized by the authors.

\ifdefined\VersionLong
	\newcommand{\CCIS}{Communications in Computer and Information Science}
	\newcommand{\ENTCS}{Electronic Notes in Theoretical Computer Science}
	\newcommand{\FAC}{Formal Aspects of Computing}
	\newcommand{\FundInf}{Fundamenta Informaticae}
	\newcommand{\FMSD}{Formal Methods in System Design}
	\newcommand{\IJFCS}{International Journal of Foundations of Computer Science}
	\newcommand{\IJSSE}{International Journal of Secure Software Engineering}
	\newcommand{\IPL}{Information Processing Letters}
	\newcommand{\JAIR}{Journal of Artificial Intelligence Research}
	\newcommand{\JLAP}{Journal of Logic and Algebraic Programming}
	\newcommand{\JLAMP}{Journal of Logical and Algebraic Methods in Programming} %
	\newcommand{\JLC}{Journal of Logic and Computation}
	\newcommand{\LMCS}{Logical Methods in Computer Science}
	\newcommand{\LNCS}{Lecture Notes in Computer Science}
	\newcommand{\RESS}{Reliability Engineering \& System Safety}
	\newcommand{\RTS}{Real-Time Systems}
	\newcommand{\SCP}{Science of Computer Programming}
	\newcommand{\SOSYM}{Software and Systems Modeling} %
	\newcommand{\STTT}{International Journal on Software Tools for Technology Transfer}
	\newcommand{\TCS}{Theoretical Computer Science}
	\newcommand{\TOPLAS}{{ACM} Transactions on Programming Languages and Systems} %
	\newcommand{\ToPNoC}{Transactions on {P}etri Nets and Other Models of Concurrency}
	\newcommand{\TOSEM}{{ACM} Transactions on Software Engineering and Methodology} %
	\newcommand{\TSE}{{IEEE} Transactions on Software Engineering}
\else
	\newcommand{\CCIS}{CCIS}
	\newcommand{\ENTCS}{ENTCS}
	\newcommand{\FAC}{FAC}
	\newcommand{\FundInf}{FI}
	\newcommand{\FMSD}{FMSD}
	\newcommand{\IJFCS}{IJFCS}
	\newcommand{\IJSSE}{IJSSE}
	\newcommand{\IPL}{IPL}
	\newcommand{\JAIR}{JAIR}
	\newcommand{\JLAP}{JLAP}
	\newcommand{\JLAMP}{JLAMP}
	\newcommand{\JLC}{JLC}
	\newcommand{\LMCS}{LMCS}
	\newcommand{\LNCS}{LNCS}
	\newcommand{\RESS}{RESS}
	\newcommand{\RTS}{RTS}
	\newcommand{\SCP}{SCP}
	\newcommand{\SOSYM}{{SoSyM}}
	\newcommand{\STTT}{STTT}
	\newcommand{\TCS}{TCS}
	\newcommand{\TOPLAS}{ToPLAS}
	\newcommand{\ToPNoC}{ToPNOC}
	\newcommand{\TOSEM}{ToSEM}
	\newcommand{\TSE}{TSE}
\fi

\bibliography{ref}
\bibliographystyle{IEEEtran}

\begin{ShortVersionBlock}
 \begin{IEEEbiography}[{\includegraphics[width=1in,height=1.25in,clip,keepaspectratio]{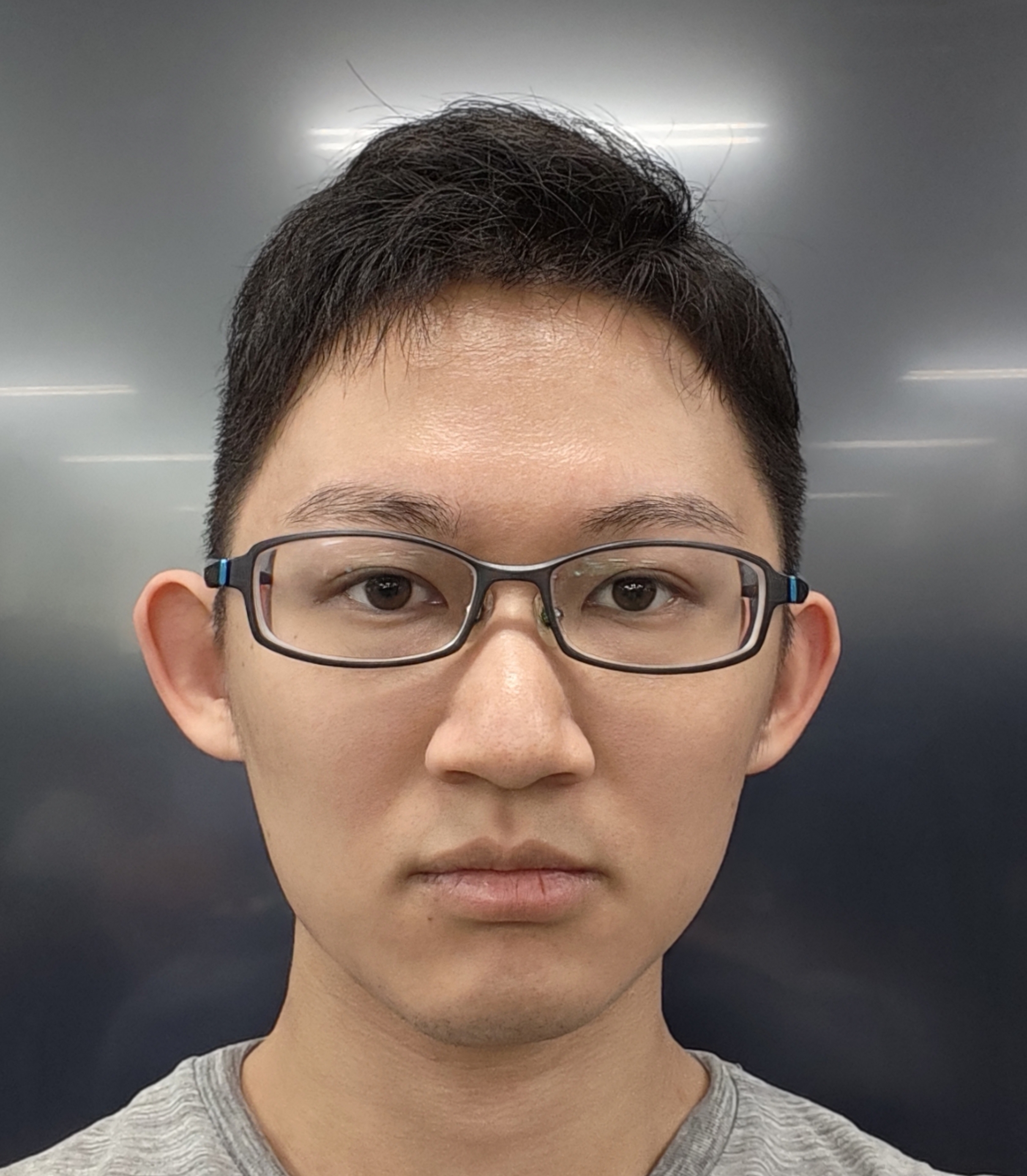}}]{Tsubasa Matsumoto} Tsubasa Matsumoto received the M.S. degree in informatics from Kyoto University, Kyoto, Japan, in 2026. He is currently working as a software engineer in the industry. \end{IEEEbiography}

 \begin{IEEEbiography}[{\includegraphics[width=1in,height=1.25in,clip,keepaspectratio]{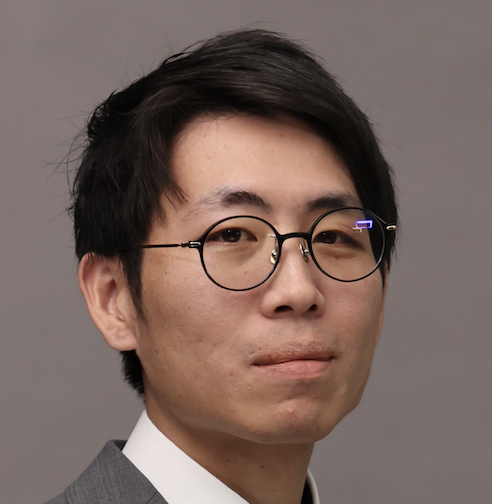}}]{Kazuki Watanabe}  Kazuki Watanabe, Ph.D., is an Assistant Professor at the National Institute of Informatics, Japan.
 He received his Ph.D.\ in Informatics from the Graduate University for Advanced Studies (SOKENDAI), Tokyo.
 His research focuses on model checking, program verification, and applied category theory. \end{IEEEbiography}

 \begin{IEEEbiography}[{\includegraphics[width=1in,height=1.25in,clip,trim=2.7cm 0cm 2.7cm 0cm,keepaspectratio]{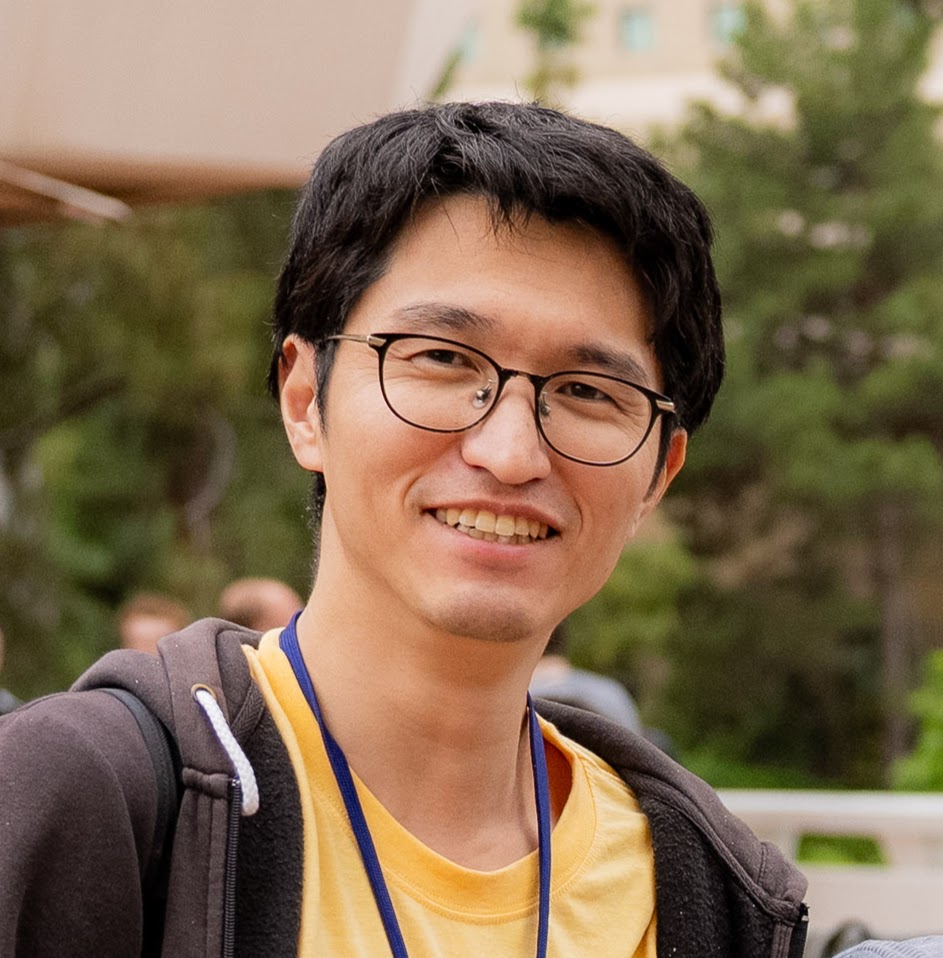}}]{Masaki Waga}
 Masaki Waga, Ph.D., is an Assistant Professor at the Graduate School of Informatics, Kyoto University, Japan.
 He received his Ph.D.\ in Informatics from the Graduate University for Advanced Studies (SOKENDAI), Tokyo, earning the Dean's Award in 2020 for his doctoral research.
 His research focuses on formal methods for cyber-physical and AI systems, spanning automata theory, automata learning, runtime verification, testing, and model checking.
 \end{IEEEbiography}
\end{ShortVersionBlock}

\begin{LongVersionBlock}

\appendix

\subsection{Rationale Behind Design Choices in \ourTool{}}

\subsubsection{Product Construction within \ourTool{}}

Although it is technically possible to export the reduced MC produced by our abstraction and give it to \storm{} with the verified LTL formula,
\ourTool{} exports the product of the reduced MC and the DRA corresponding to the given LTL formula.
This choice avoids the translation from \emph{transition-labeled} MCs to \emph{state-labeled} MCs, which increases the number of states by a factor of $|\Sigma|$ in the worst case, where $\Sigma$ is the alphabet of the DRA.
Note that \ourTool{} uses \storm{}'s DRN format as the format of MCs, which only supports state-labeled MCs.

In \revision{\ourTool{}}, we take a product of the reduced \emph{transition-labeled} MC and the DRA.
Since both have labels in $\Sigma = 2^{\AP}$ on transitions, this product construction \revision{avoids the additional $|\Sigma|$-factor state increase that would arise from translating transition labels to state labels}.
We then export a state-labeled DRN encoding of this product MC.
Moreover, we use DRAs with state-based acceptance conditions, and the labels corresponding to the acceptance condition can be naturally attached to the states of the product MC.

\subsubsection{Translation of State-labeled to Transition-labeled MCs}

Before conducting the abstraction,
we internally translate the given state-labeled MC to a transition-labeled MC.\@
There are two natural conventions for this translation: moving the label of a state to its outgoing or incoming edges.
We adopt the latter convention, because erasability is defined by the replaceability of incoming edges, and aligning the outputs on incoming edges usually makes more states erasable.

\end{LongVersionBlock}

\ifdefined\VersionWithComments
\subsection{Additional information for us, the authors}

\instructions{(EMSOFT'26-revision) You may use up to 14 pages in the IEEE TCAD format for the revised submission. Everything counts toward the page limit, including references and appendices (if any). The firm deadline for the revised version is JUNE 19, 2026, AoE.}
\todo{submit the revised paper without highlighting and provide a ``latexdiff'' version between the original submission and the revised manuscript.}
\mw{Let's use American English for spelling.}
\mw{We call the ``nodes'' of MCs ``states'', while we call those of DRAs ``positions''.}
\mw{Let's write ``$\trace_1$ and $\trace_2$ are compatible with respect to $\automaton$'' if we have $\trace_1 \sim_{\automaton} \trace_2$.}

\setcounter{tocdepth}{1}
\listoftodos{}
\fi

\ifdefined\WithReply%
	\clearpage
	\newpage

\newcounter{reviewer}
\newcounter{comment}

\setcounter{reviewer}{0}

\newcounter{questionctr}

\newenvironment{myquestion}{%
	\noindent%
	\refstepcounter{comment}%
	\textbf{Comment~\thereviewer.\thecomment}%
	\newline%
}{\par}  %

\newcommand{\questionResponse}[4]{
    \tikzstyle{reviewerQuestion} = [
     draw=black,fill=#3,text=black,
     line width=0.5pt,
     text width = \linewidth - 1.0 ex - 1pt,
     inner sep = 0.8 ex,
     rounded corners=3pt]
     {\centering
     \begin{tikzpicture}%
      \footnotesize
      \draw node[reviewerQuestion]{#4: \em #1};%
     \end{tikzpicture}}
			\noindent \textbf{Response: }
			#2\xspace\hfill$\blacksquare$

		\smallskip
}

\definecolor{boxMeta}{rgb}{.74, 0.6, 1}
\newcommand{\questionResponseMeta}[2]{\questionResponse{#1}{#2}{boxMeta}{Meta-Review}}

\definecolor{boxOne}{rgb}{0.74, 1, .74}
\newcommand{\questionResponseOne}[2]{\questionResponse{#1}{#2}{boxOne}{Reviewer 1}}

\definecolor{boxTwo}{rgb}{1, .74, .74}
\newcommand{\questionResponseTwo}[2]{\questionResponse{#1}{#2}{boxTwo}{Reviewer 2}}

\definecolor{boxThree}{rgb}{1, 1, 0.64}
\newcommand{\questionResponseThree}[2]{\questionResponse{#1}{#2}{boxThree}{Reviewer 3}}

\newcommand{\commentTitle}[1]{%
	\addtocounter{reviewer}{1}
	\setcounter{comment}{0}

    \section*{#1}

}

\newcommand{\crefpageref}[1]{\cref{#1}, page~\pageref{#1}}

\newcommand{\citeNewTextGen}[2]{
	\begin{center}

		\begin{tabular}{|p{.8\textwidth}} %
			\small{

			``#1''

			\hfill[#2]
			}
		\end{tabular}
	\end{center}
}

\newcommand{\citeNewText}[2]{\citeNewTextGen{#1}{\crefpageref{#2}}}

\normalsize

\newpage
\instructions{Prepare a response document (up to 2 pages) explaining how you have addressed the issues raised during the first review round in your revision. The response must use the same double-column IEEE TCAD format as the paper.}
\begin{center}
	{\large \textbf{Response Document}}
	\addcontentsline{toc}{chapter}{Response Document}
    \smallskip
\end{center}
\noindent
We are grateful to the reviewers for their constructive and insightful feedback.
Following their suggestions, we have substantially revised the manuscript, particularly by adding empirical evaluations that clarify the performance impact of specification-guided path shortcutting.
In the revised manuscript, major changes are highlighted in \revision{red}.

\section*{Comments from the Meta-Review}

We first address the comments from the meta-review.

\smallskip

\questionResponseMeta{%
Title: Revise to be more specific about the solution mechanism.%
}{%
We revised the title to ``Specification-Guided Path Shortcutting for Efficient Probabilistic Model Checking'' to more clearly reflect the main idea of the proposed abstraction.
We also revised the paper to use this terminology (\ie{} specification-guided path shortcutting) throughout the manuscript.%
}

\questionResponseMeta{%
Ablation: Isolate the contribution of specification-guided abstraction from confounding factors (DRA product encoding, label reduction). The current experiments do not cleanly attribute performance gains.}{%
To isolate the contribution of specification-guided path shortcutting, we added a new research question (RQ2) and conducted an ablation study.
The results are shown in \cref{table:ablation_study_results}.
In these new experiments, we introduced a ``No Abstraction'' configuration, which uses the same workflow as \ourTool{} but disables specification-guided path shortcutting.
Thus, this configuration still performs the DRA product construction and label reduction, but does not apply our proposed abstraction.

The results show that the DRA product construction and label reduction already improve performance in many cases, as ``No Abstraction'' is often faster than the \storm{} baseline.
We also found that specification-guided path shortcutting is not always beneficial: in some cases, its overhead, or the increase in the number of product transitions, can outweigh the benefit of reducing states.

At the same time, the ablation study shows that specification-guided path shortcutting can provide additional speedups when it substantially reduces the state space.
For example, on the \nand{} benchmark, ``Ours'' is significantly faster than ``No Abstraction'', reducing the end-to-end runtime from about 1,065 seconds to about 494 seconds.

In the revision, we discuss these limitations and benefits of specification-guided path shortcutting in \cref{section:experimens:ablation_study}.
}

\questionResponseMeta{%
Preprocessing: Add empirical comparison against preprocessing baselines to assess the relative contribution of specification-guided abstraction. The authors committed to this in their response.%
}{%
We thank the reviewers for pointing out the need for an empirical comparison with preprocessing baselines.
In the revision, we added a new research question (RQ3) along with experiments to evaluate the interaction between our abstraction and bisimulation minimization.
Specifically, in \cref{section:experimens:bisimulation_minimization,table:ablation_study_results}, we now report the performance of three workflows with bisimulation minimization as a preprocessing step: ``\storm{}+Bisim'', ``Ours+Bisim'', and ``No Abstraction+Bisim''.

The new results indicate that bisimulation minimization can further improve the performance of \ourTool{} in several cases, which suggests that it is complementary to specification-guided path shortcutting.
At the same time, the effect of bisimulation minimization alone is often smaller than that of our abstraction, as shown in \cref{table:benchmark_properties,table:ablation_study_results}.
For example, on $\formula^{1}_{\nand}$, bisimulation minimization does not substantially improve the \storm{} baseline, whereas specification-guided path shortcutting reduces the end-to-end time to about 493 seconds.
This reflects our observation in \cref{section:introduction}: bisimulation is property-independent and often too strong for many realistic MCs, and thus, bisimulation minimization cannot significantly reduce the state space in many cases; in contrast, specification-guided path shortcutting can remove states that are irrelevant to the particular property, which can lead to more substantial reductions.

We added \cref{section:experimens:bisimulation_minimization} accordingly and answer RQ3 as follows: bisimulation minimization can further improve the efficiency of \ourTool{}, but its influence is usually not as significant as that of specification-guided path shortcutting.
We also clarified in the discussion that the two techniques are not competing alternatives; rather, they can be combined, and our experiments indicate that applying bisimulation after specification-guided path shortcutting can be beneficial.
}

\questionResponseMeta{%
Choice of K: Provide practical guidance for selecting the MCCS bound. The current justification is insufficient. Extend the sensitivity evaluation to a broader range of K values.%
}{%
To provide more practical guidance for selecting the MCCS bound $\MCCSBound$, we extended the experiments by increasing both the number of tested properties and the range of $\MCCSBound$ values.
Since $\MCCSBound$ controls the maximum exploration depth during MCCS construction, a larger $\MCCSBound$ can, in principle, discover longer shortcuts, but it can also increase preprocessing time.
On the benchmarks we tested, increasing $\MCCSBound$ did not change the number of reduced states, while it sometimes substantially increased preprocessing time.
These results suggest that a small to moderate value, such as $\MCCSBound = 2$ or $\MCCSBound = 3$, seems to be a practical default, at least for these benchmarks.
We revised \cref{section:experimens:parameter_sensitivity} to clarify the above discussion and provide a suggestion on the choice of $\MCCSBound$.
}

\questionResponseMeta{%
Related work: Explicitly distinguish the proposed approach from prior specification-guided abstraction work (Dutreix et al; Cleaveland et al.).%
}{%
We revised \cref{section:related} to clarify how our approach differs from prior specification-dependent abstraction techniques, in particular those of Cleaveland et al.~\cite{CleavelandRSL22} and Dutreix and Coogan~\cite{DutreixC21}.
Specifically, we now clarify that the work by Cleaveland et al.\ abstracts probabilistic automata by eliminating nondeterministic branching based on the specification under a monotonicity assumption, which is orthogonal to our approach.

We also clarify that the technique developed by Dutreix and Coogan is an abstraction-refinement approach for probabilistic verification~\cite{DutreixC21}: their approach iteratively refines an abstraction based on counterexamples found during model checking, whereas our approach is a one-shot path-shortcutting transformation applied before model checking.

In the revision, we position our contribution as orthogonal to these prior methods by comparing our approach with theirs in \cref{section:related}.
}

\section*{Questions and comments by each reviewer}

Below, we provide detailed responses and clarifications to address the questions and concerns raised by the reviewers, focusing on issues not addressed during the Q\&A Week or in the meta-review response above.

\commentTitle{comments by Reviewer 2}

\questionResponseTwo{%
Line 196, I don't know how to pronounce that character and that made it a little difficult to read for me. Either use a more standard character, or tell the reader how to pronounce it.
}{%
We have revised the notation to use the character $\sigmaAlg$ for $\sigma$-algebras.
}

\questionResponseTwo{%
Definition 5: I could have used a little more intuition for BSCC: essentially (if I understand correctly) it's a subset of nodes that is strongly connected and from which you can't get out.
}{%
We agree, and we have included an intuition just before Definition 5. 
}

\questionResponseTwo{%
Lemma 6: could add more intuition: ``almost surely you end up in a BSCC''
}{%
We have included an explanation above Lemma 6. 
}

\questionResponseTwo{%
Definition 23: could add more intuition: essentially you can erase a state if you can erase all its incoming edges. Similarly Algorithm 1 just tries to erase every state one by one.
}{%
We agree, and we have added an intuition before Definition 23. There, we also explain how this notion is directly used in our algorithm.
}

\commentTitle{Comments by Reviewer 3}

\questionResponseThree{%
[page 8, left-line 574-578]: The scalability of DRA construction and its interaction with abstraction deserves more discussion. In particular, how does the approach behave for complex LTL formulas that induce large automata?
}{%
While a systematic study of DRA size is beyond the scope of this revision, the ``Prep. Time'' column for the ``No Abstraction'' configuration in \cref{table:ablation_study_results} partially reflects the cost of DRA construction and product construction.
We note that this does not isolate the execution time of these constructions because it also includes the time for serialization (\ie{} reading and writing MCs in the DRN format).
It appears that preprocessing indeed takes longer for larger LTL formulas (\eg{} $\formula^{11}_{\brp}$ and $\formula^{10}_{\haddad}$).
A systematic evaluation on larger LTL formulas would be an interesting direction for future work.
}

\questionResponseThree{%
[page 9, experiments]: From a systems perspective, it would be interesting to see integration with symbolic representations (e.g., BDD-based MCs) or extensions beyond Markov chains (e.g., MDPs).
}{%
Thank you for these suggestions. We agree that symbolic path shortcutting, \eg{} using BDD-based representations, and extensions beyond Markov chains, especially to MDPs, are important future directions.
We clarified these directions in \cref{section:related,section:conclusion}.
}

\questionResponseThree{%
[page 12]: Why not add the appendix to the main text?
}{%
In the appendix, we describe several engineering design choices made in the implementation of \ourTool{}, such as the convention used in constructing a transition-labeled MC from a state-labeled MC.\@
These details are more implementation-specific than the main technical contribution, but they may be useful for gaining a better understanding of our implementation.
For this reason, we believe the appendix is the appropriate place for these implementation details, while the main text remains focused on the core method and evaluation.
}

\fi

\end{document}